\documentclass[12pt]{article}
\pdfoutput=1 
\usepackage[english]{babel}
\usepackage[utf8x]{inputenc}
\usepackage[T1]{fontenc}
\usepackage{authblk}

\usepackage[a4paper,top=3cm,bottom=2cm,left=3cm,right=3cm,marginparwidth=1.75cm]{geometry}

\usepackage{amsmath,amsthm,amssymb} 
\usepackage{natbib}
\usepackage{enumitem}
\usepackage{graphicx}
\usepackage{caption}
\usepackage{subfig}
\usepackage{color}
\usepackage[colorinlistoftodos]{todonotes}
\usepackage[colorlinks=true, allcolors=blue]{hyperref}
\usepackage{mathrsfs}
\newtheorem{theorem}{Theorem}[section]
\newtheorem{definition}{Definition}[section]
\newtheorem{lemma}[theorem]{Lemma}

\newcommand\norm[1]{\left\lVert#1\right\rVert}
\newcommand{\by}{\mathbf{y}}
\newcommand{\bY}{\mathbf{Y}}
\newcommand{\bb}{\mathbf{b}}
\newcommand{\bR}{\mathbf{R}}
\newcommand{\bD}{\mathbf{D}}
\newcommand{\bv}{\mathbf{v}}
\newcommand{\bS}{\mathbf{S}}
\newcommand{\bW}{\mathbf{W}}
\newcommand{\bM}{\mathbf{M}}
\newcommand{\bOmega}{\mathbf{\Omega}}

\newcommand{\bB}{\mathbf{B}}
\newcommand{\bQ}{\mathbf{Q}}
\newcommand{\bbeta}{\boldsymbol{\beta}}
\newcommand{\bSigma}{\mathbf{\Sigma}}
\newcommand{\bomega}{\boldsymbol{\omega}}
\newcommand{\bepsilon}{\boldsymbol{\epsilon}}
\newcommand{\bO}{\mathbf{O}}
\newcommand{\bK}{\mathbf{K}}
\newcommand{\bV}{\mathbf{V}}
\newcommand{\bP}{\mathbf{P}}
\newcommand{\bTheta}{\mathbf{\Theta}}
\newcommand{\btheta}{\boldsymbol{\theta}}
\newcommand{\bGamma}{\mathbf{\Gamma}}
\usepackage{amssymb}
\usepackage{dcolumn}
\usepackage{hyperref}
\usepackage{bm}
\usepackage{silence}
\title{On Posterior Consistency of Bayesian Factor Models in High Dimensions}
\author[1]{Yucong Ma}
\author[1]{Jun S. Liu}
\affil[1]{Department of Statistics, Harvard University}

\begin{document}
\maketitle
\begin{abstract}
 As a principled dimension reduction technique, factor models have been widely adopted in social science, economics, bioinformatics, and many other fields. 
 However, in high-dimensional settings, conducting a `correct'  Bayesian factor analysis can be subtle since it requires both a careful prescription of the prior distribution and a suitable computational strategy.  
 In particular, we analyze the issues related to the attempt of being ``noninformative" for 
 elements of the factor loading matrix, especially 
 for sparse Bayesian factor models in high dimensions, and propose solutions to them. 
 We show here why adopting the $\sqrt{n}$-orthonormal factor assumption is appropriate and
 can result in a consistent posterior inference of the loading matrix conditional on the true idiosyncratic variance and the allocation of nonzero elements in the true loading matrix. We also provide an efficient Gibbs sampler to conduct the full posterior inference based on the prior setup from \cite{rovckova2016fast} and a uniform $\sqrt{n}$-orthonormal factor assumption on the factor matrix. 
\end{abstract}
\section{Introduction}

Factor models have been widely adopted in social science, economics, bioinformatics, and many other fields that need interpretable dimension reduction for their data. They serve as a formal way to encode 
high-dimensional observations as a linear combination of a few latent factors  plus idiosyncratic errors, which accommodate some intuitive interpretations and can sometimes be further validated by   additional  knowledge. In this article, we consider the following standard parametric formulation:  each  $G$-dimensional vector observation $\by_i$ (e.g., daily returns of $\sim$3000 U.S. stocks) is assumed to be linearly related to a  $K$-dimensional vector of latent factors $\bomega_i$ (e.g., 20 market factors) through a skinny tall factor loading matrix $\bB$:
  \begin{equation} \label{FactorModel}
  \by_i\mid \bomega_i,\bB,\bSigma\stackrel{i.i.d.}{\sim} \mathcal{N}_G(\bB\bomega_i,\bSigma), \ i=1,\ldots, n,
  \end{equation}
 and the idiosyncratic variance matrix $\bSigma$ is assumed to be diagonal as in the literature. In matrix form, we denote  the observations as $\bY=(\by_1,\cdots,\by_n)$, which is a $G\times n$ matrix, and the factors as a $K\times n$ matrix
 $\bOmega=(\bomega_1,\ldots,\bomega_n)$.
 The  factors are usually assumed to be independently and normally distributed: $\bomega_i\sim \mathcal{N}_K(\mathbf{0},\mathbf{I}_K)$. 
 
People are often interested in estimating the $G\times K$ loading matrix $\bB$ in order to gain insight on the correlation structure of the observations. 
Marginalizing out $\bomega_i$, we obtain  that  $[\by_i\mid \bB, \bSigma] \sim \mathcal{N}_G(\mathbf{0},\bB\bB^T+\bSigma),$ implying that the loading matrix $\bB$ is only identifiable up to  a right orthogonal transformation (rotationally invariant). It is thus rather difficult to pinpoint the factor loading matrix consistently, to determine the dimensionality of the  latent factors, or to design efficient algorithms to conduct a proper full Bayesian analysis of the model.


In recent years, researchers begin to investigate the effects of  sparsity assumptions on factor loadings, since a sparse loading matrix has a better interpretability and is easier to be identified. Considerable progresses have been made in the realm of sparse Bayesian factor analysis,  such as  \cite{fruehwirth2018sparse} and \cite{rovckova2016fast},
which are two representatives of the approaches using hierarchical continuous or discrete  spike-and-slab (SpSL) priors (i.e., a mixture of a concentrated distribution, which can be either continuous with a small variance or a point mass,  and a diffuse distribution) to represent the sparsity of the factor loading matrix.
The identifiability issues of sparse factor models are formally discussed in \cite{fruehwirth2018sparse}, who also designed an efficient Markov chain Monte Carlo (MCMC) procedure to simulate from the posterior distribution of an over-parameterized sparse factor model under the discrete SpSL prior. 
\cite{rovckova2016fast} proposed a sparse Bayesian factor analysis framework assuming independent (conditioned on the feature allocation) continuous SpSL priors on loading matrix's elements, under which a fast posterior mode-detecting strategy is proposed. 

Our work originates from a peculiar phenomena we observed when implementing a full Bayesian inference procedure for the factor model in (\ref{FactorModel}) under the SpSL prior from \cite{rovckova2016fast}.
Although the simulation studies of \cite{rovckova2016fast}  show a good consistency (up to trivial rotations) of the {\it maximum a posteriori} (MAP) estimation of the loading matrix in various large $G$  and large $n$ scenarios,  
we found that the corresponding Wald type consistency for the posterior distribution requires $n$ diverging at a faster rate than $s$ besides other numerical conditions on the true loading matrix that are generally required for justifying the posterior contraction   \citep{pati2014posterior}.
Here $s$ is the average number of nonzero elements of each column of the loading matrix $\bB$ and is usually much smaller than $G$.

When $s\geq n$ but is still much smaller than $G$, we observed from simulations a `magnitude inflation' phenomenon. 
That is, posterior samples of the loading matrix are inflated in the matrix norm compared to the data-generating loading matrix, and the extent of inflation is affected by the variance of the slab part of the SpSL prior ---the more diffuse the slab prior we use the more inflation we observe. This $s\geq n$ setting is not unusual in practice.  For example, the gene expression dataset analyzed in Section 7 contains measures of mRNA expression levels of $G=8932$ genes in 10 mice in four age periods ($n=40$). Each factor may correspond to a pathway and $s$ would be the average number of genes in each pathway, which can be much larger than $n$.


 The reason for this inflation phenomena is not immediately obvious since  the total number of observed quantities is $n\times G$, corresponding to $n$ observed $G$-dimensional vectors $\by_i$, $i=1,\ldots, n$, which is often much larger than $s\times K$, the number of nonzero elements in the loading matrix. Consider a special case with $K=1$, $G=s$, and $\bSigma=\mathbf{I}_G$ is known. Then 
 $\omega_i$ for $i=1,\ldots, n$ is a scalar, and $\bB=(b_1,\ldots,b_G)^T $ is a $G$-dimensional vector. Thus, each component $y_{ij}$ of $\by_i$ can be written as
 \[ y_{ij}= \omega_i b_j +\epsilon_{ij}, \ \epsilon_{ij}\stackrel{i.i.d.}{\sim} \mathcal{N}(0,1).\]
 Although the total number of unknown parameters in the model is $G+n$,  the number of independent scalar observations $y_{ij}$ is  $n\times G$, much larger than $G+n$. The model is unidentifiable because
 $\omega_i \times b_j =(\omega_i /c) \times (b_j c)$ for any $c\neq 0$. Requiring that the $\omega_i \stackrel{i.i.d.}{\sim} \mathcal{N}(0,1), i=1,\ldots, n$, can indeed alleviate the identifiability issue, but is not enough to ``tie down'' the $b_j$'s in the posterior distribution if there are too many of them, which manifests itself in the inflation phenomena. But how many is ``too many''? In this simple example,  there are ``too many'' if $G\geq n$ (Section 4).  Our later theoretical analysis shows that,  if  $s$, the column average number of nonzero elements of $\bB$, is no smaller than $n$,  the inflation will provably happen, although we observed empirically that the inflation occurs when $s \sim n$.
  An apparent remedy revealed from the above intuition and our later analysis is to further restrict the $\omega_i$'s, such as requiring that $\sum_{i=1}^n \omega_i^2 =n$. 
 
More generally speaking, due to a nearly non-identifiable structure of model (\ref{FactorModel}), an overdose of independent diffuse priors on loading matrix elements dilutes the signal  from the data. Problems with the use of diffuse priors in Bayesian inference when observation sample sizes are small relative to the number of parameters being estimated 
have been  studied  in the literature \citep{efron1973discussion,kass1996selection,natarajan1998gibbs}.
This problem for Bayesian factor analysis was also noted in \cite{ghosh2009default} and a practical solution was proposed without further theoretical investigations. 

The Ghosh-Dunson model allows each factor to have an unknown variance that follows an inverse Gamma prior and imposes the standard Gaussian prior on the loading matrix's elements. If one reallocates the variance of factors to the loading matrix side, this model is equivalent to reformatting the loading matrix as $\bB=\bQ\times\bD$ with $\bD$ being a diagonal matrix, and assuming that {\it a priori} elements in $\bQ$ are {\it i.i.d.} standard Gaussian and  diagonal elements of $\bD$ follow an inverse-Gamma distribution. When assigning non-informative priors to diagonal elements of $\bD$, elements of $\bB$ can also marginally have non-informative priors. Consequently, this hierarchical prior construction resolves the magnitude inflation problem by reducing the number of diffuse parameters, which is achieved by imposing a dependency between the magnitudes of $\bB$'s elements within the same column through $\bD$.

Some later work (e.g. \cite{Bhattacharya2011Sparse} and \cite{legramanti2020bayesian}) all follows this  loading matrix decomposition idea to induce dependencies among the magnitudes. However, informative priors are usually applied to $\bD$. In the fixed $p$ and $n\to \infty$ scenario, they develop posterior consistency results. But in the ``Large s, Small n'' scenario, these informative priors on $\bD$ can be influential for the magnitude of the loading matrix sampled from its posterior distribution as we verified in simulations.  



In this article, we  study asymptotic behaviors of the posterior distributions when an independent SpSL prior is employed for elements of the loading matrix and  a right-rotational invariant distribution is assumed on the factor matrix $\bOmega$ (i.e.,   $\bOmega$ and $\bOmega \bR$ follows the same distribution for all $n\times n$ orthogonal matrix $\bR$; this is different from the left-rotational invariance that makes $\bB$ nonidentifiable), to thoroughly understand the inflation phenomena. All consistency and convergence concepts in our work are in the frequentist (repeated-sampling) sense.
Take the loading matrix for example. If for any open neighborhood ${\cal N} $ of an entry of the true loading matrix (the magnitude of entries are at the constant order), the probability for a random draw from the posterior
distribution of that entry to fall in $\cal{N}$, as a function of the data in the repeated sampling sense,  converges to 1 almost surely as $n$ and $G$ go to infinity, we say that the posterior inference of the loading matrix is consistent, or simply that  {\it ``the posterior sample of the loading matrix converges to the truth.''}

We theoretically show that the observed inflation phenomena of the posterior distribution is due to the fact that the control of $\bOmega \bOmega^T/n$, or more specifically, the singular values of $\bOmega \bOmega^T/n$, is too weak  under only the normality assumption on the factor matrix $\bOmega$. 
Our analysis also suggests that employing a stronger control over $\bOmega \bOmega^T/n$ can result in consistent posterior distribution for the loading matrix in \cite{rovckova2016fast}'s framework under high dimensions.
Consequently, we consider the $\sqrt{n}$-orthonormal factor model: let $\bOmega/\sqrt{n}$ be uniform on the Stiefel manifold $St(K,n)$, which is the set of all orthonormal $K$-frames in ${\displaystyle \mathbb {R} ^{n}}$, or, equivalently,
the first $K$ rows of a $n\times n$ Haar-distributed random orthogonal matrix (there exists a unique right and left invariant Haar measure on the set of orthogonal matrices, see \cite{meckes2014concentration}).

From the modelling perspective, whenever the data is generated from the normal factor model~(\ref{FactorModel}) where $\bY=\bB\bOmega+\Delta$, it can also be viewed as generated by an $\sqrt{n}$-orthonormal factor model $\bY=(\bB\bK(\bOmega)/\sqrt{n})\times (\sqrt{n}\cdot\bV(\bOmega))+\Delta$ with loading matrix being $(\bB\bK(\bOmega)/\sqrt{n})$. Here $\bK(\bOmega)$ and $\bV(\bOmega)$ are from the LQ decomposition $\bOmega=\bK(\bOmega)\bV(\bOmega)$. The new loading matrix $(\bB\bK(\bOmega)/\sqrt{n})$ inherits the same generalized lower triangular structure \citep{fruehwirth2018sparse} from $\bB$ (if it posses any) and they are identical in the asymptotic sense as $n \to \infty$. Beside having the same model interpretability, we reveal in our work that the $\sqrt{n}$-orthonormal factor model enjoys two major advantages: 
\begin{enumerate}[label=(\alph*)]
    \item The posterior distribution is more robust against the choice of the  prior distribution for elements of the loading matrix in the ``Large s, Small n'' scenario. 
    The posterior consistency can hold for a broader set of prior choices including the one from \cite{rovckova2016fast}.
    \item Gibbs samplers for the normal factor model can be easily adapted to handle $\sqrt{n}$-orthonormal factors by only modifying the conditional sampling step for $\bOmega$. This modification requires negligible computational cost, but leads to a significant efficiency gain in MCMC sampling.
\end{enumerate}
For these reasons, in the high-dimensional ``Large s, Small n'' scenario, we propose to use the $\sqrt{n}$-orthonormal factor model in place of the normal factor model when doing full Bayesian inference on the population covariance matrix. Our proposed Gibbs sampler provides encouraging results in both simulations and a real data example.


The article is structured as follows. Section 2 introduces the Bayesian factor model of \cite{rovckova2016fast} 
and a corresponding basic Gibbs sampler. Under their framework, Section 3 illustrates by a synthetic example the `magnitude inflation' phenomenon of the posterior samples of the loading matrix and its dependence upon the slab prior. Section 4 provides theoretical explanations for the phenomenon. Section 5 reveals the connection between the phenomenon and the factor modeling assumption, and proposes the $\sqrt{n}$-orthonormal factor model whose posterior consistency can be guaranteed. 
By revisiting the synthetic example, Section 6 numerically verifies the consistency and robustness against prior,
 and provides a comparison between our method and alternative approaches from \cite{ghosh2009default} and \cite{Bhattacharya2011Sparse}.
Section 7 presents a real-data application. Section 8 concludes with a short discussion.

\section{Bayesian sparse factor model and inference}
\subsection{Prior settings for loading coefficient selection}\label{prior_setup}

 In order to enhance model identifiability and interpretability,  one often imposes a sparsity assumption for the loading matrix. Traditional approaches considered post-hoc rotations as well as regularization methods, see, e.g. 
 \cite{kaiser1958varimax} and \cite{carvalho2008high}. By integrating these two paradigms, \cite{rovckova2016fast} proposed a sparse Bayesian factor model framework along with a fast mode-identifying PXL-EM algorithm. In their framework, the sparsity assumption on factor loading matrix is encoded through a hierarchical SpSL prior, and we mostly follow their framework in this article.

Let $\beta_{jk}$ denote the $(j,k)^{th}$ element of the loading matrix $\bB$. We assume that {\it a priori} the $\beta_{jk}$'s follow a SpSL prior and are mutually independent given the hyper-parameters.  We introduce for each element a binary indicator  $\gamma_{jk}$ such that
 \begin{equation}\label{SSL prior}
p(\beta_{jk}|\gamma_{jk},\lambda_{0},\lambda_1)=(1-\gamma_{jk})\psi(\beta_{jk}|\lambda_{0})+\gamma_{jk}\psi(\beta_{jk}|\lambda_1), \ \ \lambda_{0}\gg \lambda_1
 \end{equation}
 where $\psi(\beta|\lambda)=\frac{\lambda}{2}\exp(-\lambda|\beta|)$ is a Laplace distribution, and 
\begin{equation}\label{IBP}
\gamma_{jk}|\theta_{k} \stackrel{ind}{\sim}  \text{Bernoulli}(\theta_{k}) \ \mbox{ and } \  \ \theta_{k}=\prod^k_{l=1} \nu_l,\  \ \nu_l \stackrel{i.i.d.}{\sim}  \text{Beta}(\alpha,1).
\end{equation} 
We note that $\theta_k$ decreases with respect to $k$.
We call $\bTheta=(\theta_1,\ldots,\theta_K)$ the {\it feature sparsity vector} and $\bGamma=(\gamma_{jk})_{G\times K}$ the {\it feature allocation} matrix.
The idiosyncratic variance matrix $\bSigma$ is assumed to be diagonal with elements $\sigma_j^2$ endowed with a conjugate prior: $\sigma_1^2,\cdots,\sigma_G^2\stackrel{i.i.d.}{\sim}  \text{Inverse-Gamma}(\eta/2,\eta \varepsilon/2)$. When $K=\infty$, the foregoing 
setup leads to an infinite factor model, for which some weak consistency results of the posterior distribution are established  for the fixed-$p$ scenario \citep{rovckova2016fast}. In simulations, they adopted a truncated approximation to the infinite factor model by setting $K$ to a pre-specified value larger than the true $K$ in data generation. Throughout the paper, we assume that $K$  is a pre-specified finite value.

\cite{rovckova2016fast} showed in simulations that the PXL-EM converges dramatically faster than the EM algorithm in finding the {\it maximum  a posteriori} (MAP) estimator (i.e., $\hat{\bB},\hat{\bSigma},\hat{\bTheta}$ that maximizes $\pi(\bB,\bSigma,\bTheta\mid\bY)$) and also demonstrated the consistency of MAP estimator in estimating the loading matrix under the ``Large s, Small n'' setting.
However, converting their method into a full Bayesian inference procedure turns out to be more subtle and challenging. 

\subsection{A standard Gibbs sampling procedure}\label{standard-Gibbs}
The full posterior distribution of the parameters, $(\bB,\bOmega,\bSigma,\bGamma,\bTheta)$, in a Bayes factor model can be written generically as
\begin{equation}\label{full_posterior}
\pi(\bB,\bOmega,\bSigma,\bGamma,\bTheta\mid \bY)\propto f(\bY|\bB,\bOmega,\bSigma)f(\bOmega) p(\bB|\bGamma)p(\bGamma|\bTheta)p(\bTheta)p(\bSigma),
\end{equation}
where $f$ denotes the likelihood, $p$ denotes prior,
$\bOmega$ denotes the $K \times n$ matrix with columns given by  $\bomega_i$, $\bGamma$ denotes the $G \times K$ matrix  with entries given by $\gamma_{jk}$ and $\bTheta$ denotes the $K$-dimensional feature sparsity vector formed by the $\theta_{k}$'s. Here observation $\bY$ represents a $G \times n$ matrix with columns $\by_i$.  

A standard 
Gibbs sampler \citep{gelfand1990sampling, liu2008monte, tanner1987calculation} for sampling from the full posterior distribution (\ref{full_posterior})
 iteratively update each component according to the following conditional distributions:

\begin{itemize}
\item Update $\bB$ iteratively as
\[
\pi(\beta_{jk}|\bbeta_{-jk},\bOmega,\bGamma,\bSigma)\propto \exp(-a_{jk}\beta_{jk}^2+b_{jk} \beta_{jk} - c_{jk}|\beta_{jk}|), \ \mbox{all } j, k;
\]
where $a_{jk}=\sum_{i=1}^n{\omega_{ik}^2}/2\sigma_j^2,b_{jk}=\sum_{i=1}^n{\omega_{ik}(y_{ij}-\sum_{l\neq k} {\beta_{jl}\omega_{il}})}/\sigma_j^2,c_{jk}=\lambda_1 \gamma_{jk}+\lambda_{0} (1-\gamma_{jk})$. 

This conditional density can be written as a mixture of two truncated normal density, and thus can be sampled efficiently.
\item Update $\bOmega$  component by component independently:
\[\bomega_i|\bB,\bSigma \sim \mathcal{N}_K((\mathbf{I}_K+\bB^T\bSigma^{-1}\bB)^{-1}\bB^T\bSigma^{-1}\by_i,(\mathbf{I}_K+\bB^T\bSigma^{-1}\bB)^{-1}), \ i=1,\ldots, n.\]
\item Update  $\bGamma$ component by component independently:
\[ \gamma_{jk} \mid \bB,\bTheta \sim \text{Bern}\left(\frac{{\lambda_1}\exp(-\lambda_1|\beta_{jk}|)\theta_{k}}{\lambda_{0}\exp(-\lambda_{0}|\beta_{jk}|)(1-\theta_{k})+{\lambda_1}\exp(-\lambda_1|\beta_{jk}|)\theta_{k}} \right), \]
for $j=1,\ldots, G; k=1,\ldots, K.$
\item Update $\bTheta$ iteratively: 
\[\theta_{k}|\bGamma,\btheta_{-k}\sim \text{Trunc-Beta}(\theta_{k+1},\theta_{k-1};\tilde{\alpha}_k,\tilde{\beta}_k)\]
where $\theta_{0}=1,\theta_{K+1}=0$ and
\begin{eqnarray*}
\tilde{\alpha}_k &=&
\begin{cases}
\#\{\gamma_{jk}=1,j=1,\cdots,G\}, \ k<K\\
\#\{\gamma_{jk}=1,j=1,\cdots,G\}+\alpha, \ k=K
\end{cases}, \\
\tilde{\beta}_k &=&
\#\{\gamma_{jk}=0,j=1,\cdots,G\}+1.
\end{eqnarray*}
Here Trunc-Beta$(a,b;{\alpha},{\beta})$ is the density proportional to $f_{Beta}(x;\alpha,\beta)I_{\{x\in [a,b]\}}$. 

\item Update $\bSigma$ along its diagonal:
\[\sigma_j^2|\bB,\bOmega \sim \text{Inverse-Gamma}\left(\frac{1}{2}(\eta+n),\frac{1}{2}(\eta \varepsilon +\sum_{i=1}^n {(y_{ij}-\bB_{j\cdot}^T\bomega_i)^2})\right)\]
where $\bB_{j\cdot}^T$ represents the $j$-th row vector of $\bB$.
\end{itemize}

Due to multimodality of the posterior distribution caused by the invariance of the likelihood function under matrix rotations (therefore only the sparsity prior can provide information to differentiate different modes) and the strong ties between the factor loading and common factors (thus making gaps among different modes very deep),  the  performance of this basic Gibbs sampler is very sticky and can only explore the neighborhood of the initial values. By initializing the sampler at some estimated mode such as the MAP estimator from the PXL-EM algorithm, however, this sampler appears to be a reasonable tool for exploring the local posterior behavior around the MAP. Indeed, more dramatic global MCMC transition moves are needed in order to have a fully functional MCMC sampler (see Appendix A).

\section{The magnitude inflation phenomenon}
\subsection{A synthetic example}\label{sec:synthetic}

To illustrate the magnitude inflation phenomenon in high dimensional sparse factor models, we 
generate a dataset from model (\ref{FactorModel})  similar to that of \cite{rovckova2016fast}, which consists of $n=100$ observations, $G=1956$ responses, and $K=5$ factors drawn from $\mathcal{N}(\mathbf{0}, \mathbf{I}_5)$. The true loading matrix is a block diagonal matrix as shown in the leftmost sub-figure of Figure~\ref{example1}, where black entries correspond to 1 and blank entries correspond to 0 (thus $s=500>n$). $\bSigma_{ture}$ is selected to be the identity matrix. With the synthetic dataset, we use the basic Gibbs sampler from section 2.2 with $\alpha=1/G,\eta=\epsilon=1, \lambda_{0}=20, \lambda_1 \in \{0.001,0.1\}$ and $K=8$, to explore the posterior distribution.  

Ten snapshots of heat-maps of $|\bB|$ in a Gibbs trajectory of 100 iterations initialized at the true value is displayed in Figure~\ref{example1}, from which we can conclude that the direction of each column vector in the loading matrix is well preserved during Gibbs iterations, whereas the absolute value of every  non-zero element increases over the iteration time and eventually stabilizes around a much larger value than the true one (about 4000 in our test setting with $\lambda_1=0.001$). As a demonstration of the inflation, Figure~\ref{example2}(a) displays the trace plot of $\log(|\beta_{1,1}|)$ with $\lambda=0.001$ and  0.1, respectively,
which also indicates the slow convergence of the basic Gibbs sampler using a small $\lambda_1$. 
The degree of inflation is influenced by the  ratio of the number of observations $n$ over  the average number  of nonzero elements of each column in the true factor loading matrix, $s$, as well as the choice of independent slab priors.
For example, when $n$ is increased from 100 to 1000, the posterior samples of the loading matrix stabilize around somewhere much closer to the true loading matrix.

\begin{figure}[!htb] 
\centering
\includegraphics[width=0.9\textwidth]{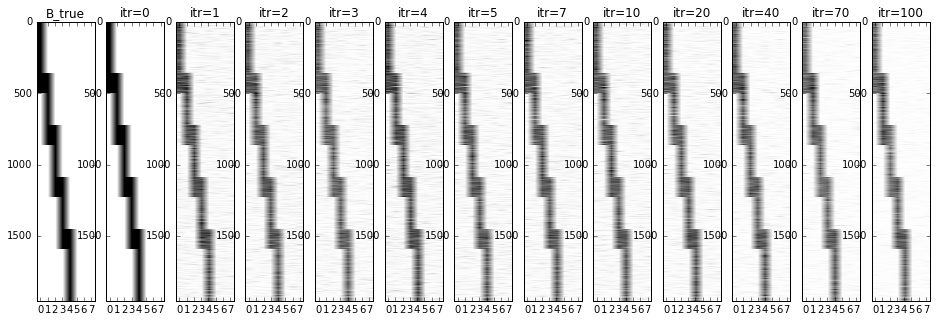} 
\caption{Heat-maps of $|\bB|$ in 100 iterations from the basic Gibbs sampler. The black entries correspond to 1 and blank entries correspond to 0. The directions of the columns of the loading matrix are well preserved throughout the Gibbs iterations.}\label{example1}
\end{figure}

\begin{figure}[!htb]
\centering
\subfloat[$n=100$]{\includegraphics[width=0.44\textwidth]{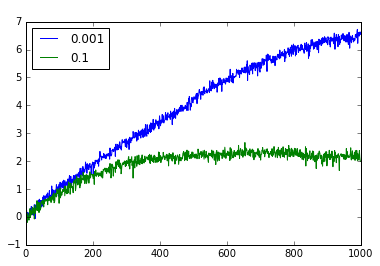}}
\subfloat[$n=1000$]{\includegraphics[width=0.45\textwidth]{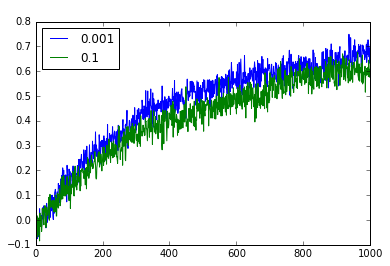}} 
\caption{Trace plot of $\log(|\beta_{1,1}|)$ from Gibbs sampler with $n=$ 100, 1000, and $\lambda_1=$  0.001, 0.1. The sampler of $\beta_{1,1}$ stabilizes around a much larger value than the truth, 1. The inflation of samples is more severe when $n$ is smaller or the variance of slab priors is larger.}\label{example2}
\end{figure}

By adding some scaling group moves \citep{Liu1999Parameter, Liu2000Generalised} to the basic Gibbs sampler (details can be found in Appendix A), which takes negligible computing time, we can greatly improve the convergence rate of the sampler, as demonstrated by contrasting Figure~\ref{example2} with Figure~\ref{example3}, of which the latter shows the trace plot for $\log(| \beta_{1,1}|)$  of the modified Gibbs sampler 
under various slab priors, for the case with $n=100$. 
Figure~\ref{example3} shows that as
$\lambda_1$ decreases from 0.5 to 0.001 so that the slab part becomes more and more diffused, the posterior mean of $|\beta_{1,1}|$ increases from around 2.5 to around 4000. Heat-maps of the factor loading are similar to Figure~\ref{example1} in all cases with $\lambda_1 \in \{0.001,0.01,0.1,0.5\}$, which means that the direction of each column vector in the loading matrix remains roughly the same throughout Gibbs iterations. 

\begin{figure}[!htb]
\centering
\includegraphics[width=0.5\textwidth]{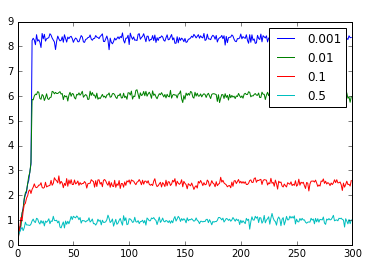}
\caption{\label{}Trace plot of $\log(|\beta_{1,1}|)$ from the modified Gibbs sampler with $\lambda_1$=0.001, 0.01, 0.1, 0.5 for the case with $n=100$. The modified Gibbs sampler has a much shorter burn-in process.}\label{example3}
\end{figure}

\subsection{Magnitude inflation and direction consistency}
Our numerical results revealed some perplexing consequences of using independent SpSL priors for a Bayesian factor model when $s \geq n$, which can be summarized as ``magnitude inflation'' and ``direction consistency''. While the former means that the posterior draws of the loading matrix are inflated entry-wise compared with the true loading matrix with the inflation magnitude dependent on how diffuse the slab prior is, the latter says that the direction of columns of posterior samples of the loading matrix somehow still converges to the true direction  as $ n, s \rightarrow \infty$. Intuitively, when the number of independent slab priors employed grows at a faster rate than the number of  observations, these priors will overwhelm the signal from data. The interesting observation is that the overdose of independent slab priors only dilutes the signal for the magnitude part in the loading matrix but has little impact on the identification of the column space. It is also worth mentioning, regardless of the occurrence of ``magnitude inflation'', the posterior distribution of the idiosyncratic variance matrix $\bSigma$ still has a nice concentration around the truth.


The inflation problem is quite a concern in practice when people try to use these posterior samples of the loading matrix for estimating the observation covariance structure. The low rank part ($\bB\bB^T$) in the estimated covariance matrix is usually exaggerated to some extent depending on the selected slab prior. Traditional literature tends to ignore the inflation problem by treating it as a consequence of the lack of  enough observations  (i.e., $n$ is too small compared to $s$) to guarantee posterior sample consistency. But this argument is inaccurate as we will show in next sections. Furthermore, we notice that, with the same amount of observations, the MAP estimator is rather precise in estimating the true loading matrix and directions of columns of the loading matrix are well captured by the posterior samples, provided that
the structure of the true feature allocation matrix is known, as in the synthetic example. 
This suggests that  the data provide sufficient information for recovering the true loading with the aid of knowing true feature allocation matrix. Thus, the magnitude inflation phenomena may be caused by some modeling issues. In the next two sections, we will provide some theoretical verifications for the magnitude inflation as well as a simple and provable remedy. 



\section{Posterior dependence on the slab prior}
It is generally recognized that in a  Bayesian factor model using an improper flat prior on elements of the loading matrix can be dangerous, and will lead to an improper posterior distribution when $G\geq n$. This is in fact not very intuitive, so we illustrate this point with a very simple example with $K=1$ factor, $n=2$ observations, and independent noises. Let the two vector observations be $\by_1$ and $\by_2$, each of $G$-dimensional. We can therefore write $\by_1= \bv_1+ \bepsilon_1$, and $\by_2 =\bv_2+\bepsilon_2$, with $\bepsilon_i \sim \mathcal{N}(\mathbf{0}, \mathbf{I}_G)$, which is very much like the canonical Normal means problem, with only one additional requirement: $\bv_1= \omega_1 \bb$ and $\bv_2 =\omega_2 \bb$. Here, the model assumes that the factor $\omega_j \sim \mathcal{N}(0,1)$, and $\bb$ is a $G$-dimensional loading matrix (vector). Thus, marginally we have $\by_i \sim N(\mathbf{0}, \mathbf{I}_G+\bb \bb^T), i=1,2$.

A peculiar thing is that in the canonical Normal means problem, if we assign flat priors to $\bv_1$ and $\bv_2$, their posterior distributions are simply $\mathcal{N}(\by_1, \mathbf{I}_G)$ and $\mathcal{N}(\by_2, \mathbf{I}_G)$, respectively, which are still proper although they yield inadmissible estimators for $\bv_1$ and $\bv_2$ when $G\geq 3$. However, with the factor model assumptions, which effectively reduce the number of parameters from $2G$ to $G$, the posterior distribution for $\bb$ becomes improper if  we assign $\bb$ a flat prior and $G\geq 2$.

Mathematically equivalent phenomena occur even in the simple univariate Gaussian mean estimation: let $y\sim \mathcal{N}(\alpha\beta, 1)$. 
If we assume that $\alpha\sim \mathcal{N}(0,1)$, then, when assuming a flat prior, the posterior distribution of $\beta$ is proportional to
$(\beta^2+1)^{-1/2}\exp\left\{-(2(\beta^2+1))^{-1}y^2\right\}$, which is a non-integrable function, thus  improper. But if we assume a proper prior on $\beta$, its posterior distribution becomes proper but its posterior variance relies heavily on its prior variance.  A simple fix of the problem is to realize that we cannot identify both parameters simultaneously and have to let $\alpha$ take a fixed value.
These phenomena also happen for the general factor models in certain settings, and our goal is to understand how these issues play out in high dimensional factor models and whether certain intuitive remedies work both theoretically and computationally for these more complex cases.  

For the general factor model, we 
can similarly marginalize out the factor variables and derive the posterior  distribution of the loading matrix under the  flat prior: 
$$\pi(\bB\mid \bY,\bSigma)\propto |\bB\bB^T+\bSigma|^{-n/2}\exp\left\{-\frac{1}{2}\text{tr}\left[(\bB\bB^T+\bSigma)^{-1}(\sum_{i=1}^n \by_i \by_i^T)\right]\right\},$$
where the exponential term is both upper and lower bounded by some functions of $\bY$ and $\bSigma$. Term $ |\bB\bB^T+\bSigma|^{-n/2}$ is lower bounded by $(||\bB||_F^2+\lambda_{max}(\bSigma))^{-\frac{n\times K}{2}}$, where $||\bB||_F$ represents the Frobenius norm of $\bB$, and $\lambda_{max}(\bSigma)$ denotes the largest eigenvalue of $\bSigma$. When the dimension of $\bB$, which is $G\times K$, is no smaller than $n\times K$, $\pi(\bB|\bY,\bSigma)$ will integrate to infinity in the complement region of any bounded set in $\mathcal{R}^{G\times K}$, 
leading to an improper posterior distribution. If we impose a proper but diffuse slab prior instead of the improper flat prior on elements of $\bB$, 
the posterior distribution can still be very sensitive to the variance of slab prior, as seen in Figure~\ref{example3}.

To formalize this intuition for general Bayesian factor models, we provide the following theorem on the divergence of the posterior distribution of the loading matrix if we use a sequence of increasingly diffused ``slab'' priors.  Note that for  theorems in Section 4, we do not require $\bSigma$ to be diagonal. To cover generic prior choices, we replace (\ref{SSL prior}) with  
\begin{equation}\label{general SSL}
p(\beta_{jk}|\gamma_{jk})=(1-\gamma_{jk})\psi(\beta_{jk})+\gamma_{jk}\phi(\beta_{jk})
 \end{equation}
 where $\psi$ denotes the spike prior density and $\phi$ denotes the slab prior density.
 

\begin{theorem}
Let $\{\phi_m\}_{m=1,\cdots}$ be a sequence of densities such that $lim_{m\rightarrow\infty} \phi_m(\beta)=0$ for every $\beta \in \mathcal{R}$ and there exists a constant $C \in (0,1)$ such that $\phi_m(\beta)>C \ max_{\beta}(\phi_m(\beta))$ holds for every $\beta$ in some non-decreasing Borel sets $S_m$ that converges to $\mathcal{R}$ as $m \rightarrow \infty$. If $s=||\bGamma||^2_F/K \geq n$, then for any fixed finite-measure Borel set $S$, $\lim_{m\rightarrow\infty} P(\bB \in S|\bY,\bSigma,\bGamma,m)=0$, where $[\bB\mid\bY,\bSigma,\bGamma,m]$ is based on the posterior distribution from model (\ref{FactorModel}) with normally distributed factors and $\phi_m$ as the slab part in the SpSL prior on loading matrix elements.
\end{theorem}

Theorem 4.1 partially explains the magnitude inflation and the dependence of the inflation rate on the choice of the slab prior. Let $S$ be any fixed $G \times K$ dimensional ball. The theorem implies that the probability of a posterior sample $\bB$, conditional on $\bY,\bSigma,\bGamma,m$, having a matrix norm smaller than any constant goes to zero as we use a series of slab priors $\{\phi_m\}_{m=1,2,\cdots}$ that is  increasingly diffused. In a general sense, it can also be understood as the convergence in distribution of $\bB|\bY,\bSigma,\bGamma,m$ towards $\bB|\bY,\bSigma,\bGamma,\infty$ (conditional posterior of $B$ with flat slab prior), which is a point mass at infinity when $s\geq n$. For cases such that $\bB|\bY,\bSigma,\bGamma,\infty$ is indeed proper, e.g., when $s \ll n$ or the assumed distribution on the factors is changed, we strictly have the convergence of $\bB|\bY,\bSigma,\bGamma,m$ towards $\bB|\bY,\bSigma,\bGamma,\infty$ in distribution as stated in the next theorem. Therefore, if the posterior distribution of the loading matrix is proper under a flat slab prior and the Bayesian consistency is justified in this situation,  we have approximately the same consistency when employing a reasonably diffuse slab prior. 

\begin{theorem}
Consider model (\ref{FactorModel}) without the normality assumption for the factors. Let $\{\phi_m\}_{m=1,\cdots}$ be a sequence of prior densities maximized at 0 such that, $\forall \beta\in \mathbb{R}$, $\lim_{m \rightarrow\infty} \phi_m(\beta)\phi^{-1}_m(0) = 1$. Let $\pi(\bB|\bY,\bSigma,\bGamma,m)$ denote the conditional posterior density of $\bB$ under a SpSL prior for its elements, with the spike density $\psi$ and the slab density $\phi_m$, and let $\pi(\bB|\bY,\bSigma,\bGamma,\infty)$ be the 
one corresponding to the flat slab prior (this is appropriate since  the indicator matrix $\bGamma$ is conditioned on). 
If $\pi(\bB|\bY,\bSigma,\bGamma,\infty)$ is integrable, then 
$\bB|\bY,\bSigma,\bGamma,m$ converges to $\bB|\bY,\bSigma,\bGamma,\infty$  in distribution as $m \rightarrow \infty$. 
\end{theorem}

\section{Model modifications and posterior consistency}
To concentrate on the magnitude inflation and direction consistency problems,
we study behaviors of the posterior distribution of the Bayesian factor model assuming that the diagonal idiosyncratic covariance 
matrix $\bSigma$ and the true number of factors (for the basic factor model) or the true feature allocation matrix $\bGamma$ (for the sparse factor model) are known. 
In contrast to the solution provided by \cite{ghosh2009default}, which imposes dependency among the magnitudes of loading matrix elements within the same column through prior setup,
we restrict ourselves to a special class of SpSL priors for loading matrix elements, which have 
a  point mass at zero as the spike and a flat (limit of a sequence of increasingly diffused distributions)  slab part. This is a natural choice for being non-informative and is always appropriate when considering the conditional posterior distributions given $\bGamma$.
We focus on studying the connection between posterior consistency and the factor assumption, and demonstrate why $\sqrt{n}$-orthonormal factor model is a natural choice under high dimensions.




\textbf{Notations:} 
 Let $H_n$ denote the Haar measure (i.e., uniform distribution) on the space of $n\times n$ orthogonal matrices and let $m_n$ be the uniform measure on the Stiefel manifold $St(K,n)$. 
 Let  $\bM_{i \cdot}$ and $\bM_{\cdot j}$ denote the $i$-th row and the $j$-th column of  matrix $\bM$, respectively, as column vectors, and let $\bM_{i,j}$ denote the element at $i$-th row and $j$-th column of $\bM$. $\bM_{i_1:i_2}$ denotes the sub-matrix formed by row $i_1$-th to $i_2$ and  $\bM_{i_1:i_2,j_1:j_2}$ denote the sub-matrix formed by rows $i_1$-th to $i_2$ and columns $j_1$ to $j_2$. Notation $\bM^{\bot}$ represents an orthogonal complement (not unique) of $\bM$ when $\bM$ is not a square matrix, $\mathcal{P}_{(\cdot)}$ represents the projection mapping towards the row vector space of a matrix and $\bP_{(\cdot)}$ is the projection matrix of the mapping. Let
 $\lambda_{max}(\cdot)$ and $\lambda_{min}(\cdot)$ denote the largest and smallest singular values of a matrix, and let $\lambda_{k}(\cdot)$ denote the $k$-th largest singular values. The $L_2$ norm is denoted by $\norm{\cdot}$, the Frobenius norm is denoted by $\norm{\cdot}_F$, and the outer product is 
  ``$\otimes$".

\subsection{The basic Bayesian factor model}
We show the posterior consistency of the loading matrix by first studying the posterior consistency of the factor matrix $\bOmega$ (defined in section 2.2). 
It is easy to see that, with a flat prior on every element of $\bB$, the  posterior distribution of $\bB$ and $\bOmega$ can be written as:
\begin{equation}\label{posterior_B}
\bB_{j\centerdot }|\bY,\bOmega, \bSigma \stackrel{ind}{\sim}
\mathcal{N}((\bOmega\bOmega^T)^{-1}\bOmega \bY_{j\centerdot},\sigma_j^2(\bOmega\bOmega^T)^{-1})
\end{equation}
\begin{equation}\label{posterior_Omega}
\pi(d\bOmega|\bY,\bSigma)\propto \ |\bOmega\bOmega^T|^{-G/2}\ \exp\left(\sum_{j=1}^G \frac{1}{2\sigma_j^2}\bY_{j\centerdot}^T\bOmega^T(\bOmega\bOmega^T)^{-1}\bOmega \bY_{j\centerdot}\right)p_{\bOmega}(d\bOmega)
\end{equation}
where $p_{\bOmega}$ denotes the prior distribution of $\bOmega$ and ``$\stackrel{ind}{\sim}$" means that the $\bB_{j\cdot}$'s  are mutually independent. 

For this section, we no longer restrict the factors in $\bOmega$ to follow the standard Normal distribution, only requiring its distribution $p_{\bOmega}$  to satisfy the following two conditions: (a) $cov(\bomega_i)=\mathbf{I}_K$, so as to keep the marginal covariance structure of $\bY$ unchanged;
(b) right rotational-invariant (i.e.,  $\bOmega$ and $\bOmega \bR$ follow the same distribution $\forall$ $n\times n$ orthogonal matrix $\bR$). 
Two non-Gaussian examples are: (i) each row of $\bOmega$ follows independently a uniform distribution on the $\sqrt{n}$-radius sphere; (ii) $\bOmega/\sqrt{n}$ is uniform on the Stiefel manifold $St(K,n)$, i.e., $\bOmega/\sqrt{n}$ is the first $K$ rows of a Haar-distributed $n \times n$ orthogonal random matrix. A straightforward characterization of condition (b) can be made through the LQ decomposition (the transpose of the QR decomposition). Suppose the LQ decomposition of $\bOmega=\bK(\bOmega)\bV(\bOmega)$ is done by 
the Gram–Schmidt orthogonalization starting from the first row of $\bOmega$,
resulting in a $K\times K$ lower triangular matrix $\bK(\bOmega)$ and a $K \times n$ orthonormal matrix $\bV(\bOmega)$. Then,  requirement (b) enables us to generate $\bOmega$ from  $p_{\bOmega}$ by generating a pair of $\bK(\bOmega)$ and $\bV(\bOmega)$ from two independent distributions---a marginal distribution on $\bK(\bOmega)$ (denoted  as $p_{\bK}$) and a uniform distribution on the Stiefel manifold $St(K,n)$ for $\bV(\bOmega)$.

Using the LQ decomposition, we can rewrite expression (\ref{posterior_Omega})  as \begin{equation}\label{posterior_Omega_}
\begin{split}
\pi(d\bOmega|\bY,\bSigma)&\propto \Big(|\bK(\bOmega)\bK(\bOmega)^T|^{-G/2}p_{\bK}(d\bK(\bOmega))\Big)\\
&\times \Big(\text{exp}\Big(\sum_{j=1}^G\frac{1}{2\sigma_j^2}\|\mathcal{P}_{\bV(\bOmega)}(\bY_{j\centerdot})\|^2\Big)m(d \bV(\bOmega))\Big)
\end{split}
\end{equation}
since $|\bOmega\bOmega^T|=|\bK(\bOmega)\bK(\bOmega)^T|$, and $\bY_{j\centerdot}^T\bOmega^T(\bOmega\bOmega^T)^{-1}\bOmega \bY_{j\centerdot}$ is the square of the length of $\bY_{j\centerdot}$'s projection on the row space of $\bOmega$.  Therefore,  $\bK(\bOmega)$ and $\bV(\bOmega)$ are independent {\it a posteriori}, and
\begin{equation}\label{posterior_K}
\pi(d\bK(\bOmega)|\bY,\bSigma) \propto |\bK(\bOmega)\bK(\bOmega)^T|^{-G/2} p_{\bK}(d\bK(\bOmega))
\end{equation}
\begin{equation}\label{posterior_V}
\pi(d \bV(\bOmega)|\bY,\bSigma)\propto \text{exp}\Big(\sum_{j=1}^G\frac{1}{2\sigma_j^2}\|\mathcal{P}_{\bV(\bOmega)}(\bY_{j\centerdot})\|^2\Big)m(d \bV(\bOmega)).
\end{equation}

Equation (\ref{posterior_K}) implies that $\bK(\bOmega)$ may have an improper posterior distribution because the likelihood term $|\bK(\bOmega)\bK(\bOmega)^T|^{-G/2}$  creates ``attractors'' when the determinant of $\bK(\bOmega)\bK(\bOmega)^T$ is close to  0. Therefore, with large enough $G$, the right-hand side of 
 (\ref{posterior_K}) explodes to infinity fast enough  around the attractors and becomes non-integrable, thus leading to an improper posterior distribution for $\bK(\bOmega)$. In contrast, since $ \text{exp}\left(\sum_{j=1}^G\frac{1}{2\sigma_j^2}\|\mathcal{P}_{\bV(\bOmega)}(\bY_{j\centerdot})\|^2\right)$ is upper bounded by  $\text{exp}\left(\sum_{j=1}^G\frac{1}{2\sigma_j^2}\|\bY_{j\centerdot}\|^2\right)$, the posterior distribution (\ref{posterior_V}) for $\bV(\bOmega)$ is always proper, based on which we can further derive posterior consistency of the row vector space of $\bOmega$. 

\subsubsection{Consistency of the row vector space of the factor matrix}
 The consistency of row vector space of $\bOmega$ is intuitive from (\ref{posterior_V}) for the noiseless case (i.e., $\bY=\bB_0 \bOmega_0$), since the exponential term in (\ref{posterior_V}) is uniquely maximized when the row vector spaces of $\bOmega$ and $\bOmega_0$ coincide. As in an annealing algorithm, the exponential term enforces the growing contraction towards the maximum point (where row spaces of $\bOmega$ and $\bOmega_0$ coincide) as $G$ increases. 
 On the other hand, the prior measure in a neighborhood of the row vector space of $\bOmega_0$ (defined as $p_{\bOmega}(\{\bOmega:||\bV(\bOmega_0)^{\bot}\bV(\bOmega)^T||_F<\epsilon\})$)  gets more diffused as $n$ grows. Therefore, in an asymptotic regime with $G,n \rightarrow \infty$, and under some mild conditions on the growing rate of $G$ and $n$ to ensure that the diffusion is slower than the contraction, the consistency of the  row vector space of $\bOmega$ follows immediately as summarized below. Detailed proofs of the lemma and theorem can be found in Appendix C.3 and C.4. 
 
 \begin{lemma}
 Let $\bB_{0,G}$ be a $G \times K$ matrix, $\bOmega_{0,n}$ be a $K \times n$ matrix, and $\bSigma_G$ be a known $G\times G$ diagonal matrix. Suppose  noiseless data generated as $\bY=\bB_{0,G} \bOmega_{0,n}$ are given. We, however,  model each column of $\bY$  as mutually independent and $\bY_{\cdot i} \sim \mathcal{N}_G(\bB \bOmega_{\cdot i},\bSigma_G)$,  $i=1,\cdots,n$. With a flat prior on each of $\bB$'s elements and a right-rotational invariant prior on $\bOmega$, we have the following inequality for the posterior distribution of $\bOmega$:
 \[
 \begin{split}
 &P(||\bV(\bOmega_{0,n})^{\bot}\bV(\bOmega)^T||_F>\epsilon|\bY, \bSigma_G)\\
 \leq  \Big(1+m_n(\{\bV: ||&\bV_0 \bV^T||_F < \frac{\epsilon}{L}\})
 \times \exp(\frac{3}{8}\epsilon^2 \lambda_{min}(\bSigma_G^{-1/2}\bB_{0,G}\bK(\bOmega_{0,n}))^2)\Big)^{-1}
 \end{split}
 \]
 where $L=2\lambda_{max}(\bSigma_G^{-1/2}\bB_{0,G}\bK(\bOmega_{0,n}))/\lambda_{min}(\bSigma_G^{-1/2}\bB_{0,G}\bK(\bOmega_{0,n}))$ and $\bV_0$ is any fixed $K\times n$ orthonormal matrix.
 \end{lemma}

Lemma 5.1 provides a probability bound between $\bV(\bOmega)$ sampled from the posterior distribution and $\bV(\bOmega_{0,n})$ when there is no noise in the observation $\bY$. Since $||\bV(\bOmega_{0,n})^{\bot}\bV(\bOmega)^T||^2_F$ equals to the sum of squared sine canonical angles between the row space of $\bOmega$ and $\bOmega_0$, lemma 5.1 implies the convergence of these canonical angles towards 0  as $n,G=s \rightarrow \infty$ (i.e. the Bayesian consistency of row vector space of $\bOmega$) when $-\log(m_n(\{\bV: ||\bV_0^{\bot} \bV^T||_F < \frac{\epsilon}{L}\}))=o(\epsilon^2 \lambda_{min}(\bSigma_G^{-1/2}\bB_{0,G}\bK(\bOmega_{0,n}))^2)$, which is 
the technical requirement that ensures the dilution is ``covered up'' by the contraction.
Base on this lemma, we generalize the consistency of row vector space of $\bOmega$ to the noisy observation case under the "Large p(s), Small n" paradigm.

\begin{definition}
Let $\bB_0$ be a countable array, or a bivariate function of the form $\bB_0(j,k)$, with $j=1,\cdots,\infty$ and $k=1,\cdots,K$. Intuitively, this is an $\infty \times K$ matrix. We say that $\bB_0$ is a regular infinite loading matrix if there are two universal constants $C_1, C_2>0$ such that, $\|(\bB_0)_{j\cdot}\|\leq C_1$ and $\lambda_{min}((\bB_0)_{1:j})/\sqrt{j} \geq C_2$ for $j=1,\cdots,\infty$.
\end{definition}

\begin{theorem}\label{thm5.2}
Suppose $\bB_0$ is a regular infinite loading matrix. Let $\bOmega_{0,n}$ be a $K \times n$ matrix with linear independent rows and and let $\bSigma=diag(\sigma_1^2,\cdots)$ be a known infinite diagonal matrix  in which $\sigma_j$, $\forall j$, is bounded below and above by constants $c_3>0$ and $c_4<\infty$, respectively. Let $\bY$ 
be an $\infty \times n$ matrix, whose $j$-th row is generated from $\mathcal{N}_n((\bB_0)_{j\cdot}\bOmega_{0,n},\sigma_j^2 \mathbf{I}_n)$, independently. For every fixed $G$, consider modeling the $i$-th column of $\bY_{1:G}$ by $\mathcal{N}_G(\bB\bOmega_{\cdot i},\bSigma_G)$ for $i=1,\dots,n$ with $\bSigma_G=diag(\sigma_1^2,\cdots,\sigma_G^2)$. With  a flat prior on each of $\bB$'s elements and a proper right-rotational invariant  prior on $\bOmega$, 
we have, for a random draw $\bOmega$ from its
 posterior distribution,   almost surely (with respect to the randomness in $\bY$) that 
 \[ ||\bV(\bOmega_{0,n})^{\bot}\bV(\bOmega)^T||_F \mid \bY_{1:G},\bSigma_G  \rightarrow 0 \mbox{ in probability as } \ G\rightarrow \infty.\]
\end{theorem}


\subsubsection{Posterior distribution of the loading matrix}
From (\ref{posterior_V}), it is clear that data only provide information on the row vector space of $\bV(\bOmega)$, the posterior distribution of $\bV(\bOmega)$ conditioned on its row vector space is uniform among all the $K \times n$ orthonormal matrices within the row space. Utilizing the posterior consistency of the row space provided by Theorem~\ref{thm5.2}, we can approximate 
an $\bV(\bOmega)$ drawn from its posterior by another random variable of the form $\bO \bV(\bOmega_{0,n})$, where $\bO$ is a $K \times K$ uniform (Haar distributed) random orthogonal matrix (see Appendix C.5 for details).
  
Let $\bB_{0,G}$ denotes the matrix formed by the first $G$ rows of $\bB_0$. By plugging $\bV(\bOmega)=\bO \bV(\bOmega_{0,n})$ into the matrix form of (\ref{posterior_B}), which can be written as
\[\bB\mid \bY, \bOmega, \bSigma \sim \mathcal{N}_{K \times G}(\bY\bOmega^T (\bOmega\bOmega^T)^{-1} ,  (\bOmega\bOmega^T)^{-1}\otimes\bSigma),\]
we obtain a decomposition for the posterior samples of $\bB\bK(\bOmega)/\sqrt{n}$ as:
\begin{eqnarray}
\frac{1}{\sqrt{n}}\bB\bK(\bOmega)\mid \bY,\bSigma &\sim & \bB_{0,G}(\bK(\bOmega_{0,n})/\sqrt{n})\bO^T+((\bY-\bB_{0,G}\bOmega_{0,n})/\sqrt{n}) \bV(\bOmega_{0,n})^T \bO^T
\nonumber\\
& & +\mathcal{N}_{G \times K}(\mathbf{0},\frac{1}{n} \mathbf{I}_K \otimes \bSigma ).\label{posterior_BK}
\end{eqnarray}
For a considerable large $n$ and normal true factor matrix $\bOmega_{0,n}$, $\bK(\bOmega_{0,n})/\sqrt{n}$, as the Cholesky factor of  $\bOmega_{0,n}\bOmega_{0,n}^T/n$, approaches the identity matrix, so the first term of the right hand side of (\ref{posterior_BK}) approaches $\bB_{0,G}\bO^T$.
Meanwhile, the second term $((\bY-\bB_{0,G}\bOmega_{0,n})/\sqrt{n}) \bV(\bOmega_{0,n})^T \bO^T$ is the row projection of the idiosyncratic noise matrix $(\bY-\bB_{0,G}\bOmega_{0,n})$ to a $K$ dimensional space, divided by $\sqrt{n}$, which converges in probability to 0 entry-wise as $n \rightarrow \infty$. The third term is a centered normal (independent with $\bO$) with variance shrinking to 0 as $n$ increases. This implies that under $G=s \gg n \rightarrow \infty$ regime, posterior samples of $\bB\bK(\bOmega)/\sqrt{n}$ can be asymptotically expressed as the true loading matrix times an uniform random orthogonal matrix. 

{\bf Factor assumption and consistency.} Posterior distributions of $\bB$ and $\bK(\bOmega)$ are coupled. A ``deflation'' problem of $\bK(\bOmega)/\sqrt{n}$ occurs when the factors in $\bOmega$ are assumed to be normal and $n=O(G)$, in which case the posterior distribution of $\bK(\bOmega)/\sqrt{n}$ can be derived in closed form by the Bartlett decomposition as:
\begin{equation}
\begin{split}
\frac{1}{\sqrt{n}}(\bK(\bOmega))_{k,k}|\bY,\bSigma \sim \frac{1}{\sqrt{n}}\chi_{n-k+1-G}, \ k=1,\cdots,K,\\
\frac{1}{\sqrt{n}}(\bK(\bOmega))_{k',k}|\bY,\bSigma \sim \mathcal{N}(0,\frac{1}{n}),\ 1\leq k < k' \leq K,
\end{split}
\end{equation}
where $\chi_\nu$ denotes the Chi distribution with $\nu$ degrees of freedom.
Posterior samples of the loading matrix, therefore, have to be  inflated correspondingly. Ideally, we desire the convergence of the posterior distribution of $\bK(\bOmega)/\sqrt{n}$ towards a point mass at the identity matrix to guarantee the posterior consistency (up to rotations) of the loading matrix, and can indeed achieve this by imposing a stronger control over the singular values of $\bOmega$ through the assumption on $p_{\bOmega}$. Such remedy is not unique. A particular simple strategy is to require that all factors are orthogonal and have equal norm, 
which implies that  $\bOmega/\sqrt{n}$ is uniform in the Stiefel manifold $St(K,n)$. 
More discussions are deferred to Section 5.2.2.


\subsection{Sparse Bayesian factor model}
With a special feature allocation design, $\bV(\bOmega)$ is identifiable so that the consistency of the row space of the factor matrix can be generalized to the consistency of $\bV(\bOmega)$. We impose a {\it generalized lower triangular structure} \citep{fruehwirth2018sparse} on the feature allocation matrix $\bGamma$ to cope with the rotational invariance problem of the loading matrix.
We call $\bGamma$ a generalized lower triangular matrix if the row index of the top nonzero entry in the $k$-th column $l_k$ (define $l_0=1$, $l_{K+1}=G+1$) increases with $k$ and $\gamma_{jk}=1$ if and only if $j \geq l_k$.
Under the flat SpSL prior (use a mixture of point mass at zero and flat distribution as prior) on entries of $\bB$ in the Sparse Bayesian factor model introduced in section 2.1, we can derive the conditional distributions of $\bB$ and $\bOmega$: for $j=l_k,\cdots,l_{k+1}-1$,
\begin{equation}\label{posterior_B2}
\bB_{j, 1:k}|\bY,\bOmega, \bSigma,\bGamma \stackrel{ind}{\sim} \mathcal{N}((\bOmega_{1:k}{\bOmega_{1:k}}^T)^{-1}\bOmega_{1:k} \bY_{j\centerdot}\ ,\ \sigma_j^2(\bOmega_{1:k}{\bOmega_{1:k}}^T)^{-1}), 
\end{equation}
\begin{equation}\label{posterior_Omega3}
\pi(d\bOmega|\bY,\bSigma,\bGamma)\propto \prod_{k=1}^K |\bOmega_{1:k}{\bOmega_{1:k}}^T|^{-(l_{k+1}-l_k)/2} \exp\left(\sum_{k=1}^K \sum_{j=l_k}^{l_{k+1}-1}\frac{1}{2\sigma_j^2}\|\mathcal{P}_{\bOmega_{1:k}}(\bY_{j\centerdot})\|^2\right)p_{\bOmega}(d\bOmega),
\end{equation}
where $\bB_{j,1:k}=(\beta_{j1},\beta_{j2},\cdots,\beta_{jk})^T$ and $p_{\bOmega}$ denotes the distribution assumed on $\bOmega$ such that condition (a) and (b) holds.

Given the LQ decomposition $\bOmega=\bK(\bOmega)\bV(\bOmega)$ and \[\bOmega_{1:k}=\bK(\bOmega)_{1:k}\bV(\bOmega)=\bK(\bOmega)_{1:k,1:k}\bV(\bOmega)_{1:k},\]
since $\bK(\bOmega)$ is lower triangular,  $\bOmega_{1:k}{\bOmega_{1:k}}^T=\bK(\bOmega)_{1:k,1:k}{\bK(\bOmega)_{1:k,1:k}}^T$ is a function of $\bK(\bOmega)$. $\mathcal{P}_{\bOmega_{1:k}}(\bY_{j\centerdot})$ is the projection of $\bY_{j\centerdot}$ towards the row vector space of $\bOmega_{1:k}$, which is  a function of $\bV(\bOmega)$. The adoption of the generalized lower triangular structure on feature allocation matrix ensures a separation in likelihood of (\ref{posterior_Omega3}) so that the determinant part is connected to $\bOmega$ only through $\bK(\bOmega)$ and the exponential part only through $\bV(\bOmega)$. We thus can derive that $\bK(\bOmega)$ and $\bV(\bOmega)$ are independent {\it a posteriori} and that:
\begin{equation}\label{posterior_K2}
\pi(d\bK(\bOmega)|\bY,\bSigma,\bGamma) \propto \prod_{k=1}^K \bK(\bOmega)_{k,k}^{-(G-l_k+1)} p_K(d\bK(\bOmega))
\end{equation}
\begin{equation}\label{posterior_V2}
\pi(d \bV(\bOmega)|\bY,\bSigma,\bGamma)\propto\exp\left(\sum_{k=1}^K \sum_{j=l_k}^{l_{k+1}-1}\frac{1}{2\sigma_j^2}\|\mathcal{P}_{\bV(\bOmega)_{1:k}}(\bY_{j\centerdot})\|^2\right)
m(d \bV(\bOmega)).
\end{equation}
Expression (\ref{posterior_V2}) gives a proper posterior for $\bV(\bOmega)$, and for the noiseless case (i.e. $\bY=\bB_0 \bOmega_0$), the density is maximized when the row vector space of $\bV(\bOmega)_{1:k}$ and $\bV(\bOmega_0)_{1:k}$ coincide for $k=1,\cdots,K$, based on which we can generalize theorem 5.2 to the consistency (up to sign permutations) of $\bV(\bOmega)$.

\subsubsection{Consistency of \texorpdfstring{$\bV(\bOmega)$}{TEXT}}
\begin{definition}
Let $\bB_0$ be an $\infty \times K$ matrix with nonzero rows and let $\bGamma_0$ be a binary matrix of the same shape. We call $\bGamma_0$  a generalized lower triangular feature allocation matrix of $\bB_0$ if it satisfies 
\begin{enumerate}
    \item $\mathbb{I}_{(\bB_0)_{j,k}\neq 0} \leq (\bGamma_0)_{j,k}$ holds for $j=1,\cdots,\infty$, $k=1,\cdots,K$, where $\mathbb{I}$ is the indicator function;
    \item $(\bGamma_0)_{j,k_1}\leq (\bGamma_0)_{j,k_2}$ holds for $j=1,\cdots,\infty$, $K\geq k_1>k_2\geq 1$.
\end{enumerate}
  Furthermore, for every fixed dimension $G$, let $\psi_G$ denote the unique permutation of $(1,\cdots,G)$, so that
$\psi_G(j_1)< \psi_G(j_2)$ if and only if either (i)
$\left(\sum_{k} \bGamma_{j_1,k}\right) <\left(\sum_k \bGamma_{j_2,k}\right)$
or (ii) $\left(\sum_{k} \bGamma_{j_1,k}\right) =\left(\sum_k \bGamma_{j_2,k}\right)$
but $j_1<j_2$. 
\end{definition}

\begin{definition}
Let $\bB_0$ be a $\infty \times K$ matrix with nonzero rows and  let $\bGamma_0$ be a generalized lower triangular feature allocation matrix of $\bB_0$. The two  $G \times K$ matrices $\bB_{0,G}$ and $\bGamma_{0,G}$ are formed by permuting the first $G$ rows of $\bB_0$ and $\bGamma_0$ according to $\psi_G$ (the $j$-th row of $\bB_0$ is the $\psi_G(j)$-th row of $\bB_{0,G}$). Let $l_{0,k}$ be the row index of the top nonzero entry in the k-th column of the generalized lower triangular matrix $\bGamma_{0,G}$ (define $l_{0,0}=1$, $l_{0,K+1}=G+1$), and let $\bB_{0,G}^{(k)}$ be the submatrix of $\bB_{0,G}$ formed by rows indexed from $l_{0,k}$ to $l_{0,k+1}-1$ and columns indexed from 1 to k. We call $(\bB_0,\bGamma_0)$ a regular infinite loading pair if there are two universal constants $C_1, C_2>0$ such that, $\|(\bB_0)_{j\cdot}\|\leq C_1$ and $min_{k} \lambda_{min}(\bB_{0,j}^{(k)})/\sqrt{j} \geq C_2$ for $j=1,\cdots,\infty$.
\end{definition}

\begin{figure}[!htb]
\centering
\includegraphics[width=0.1\textwidth]{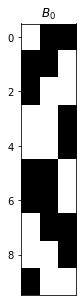}
\includegraphics[width=0.1\textwidth]{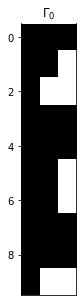}
\hspace{0.5cm}
\includegraphics[width=0.1\textwidth]{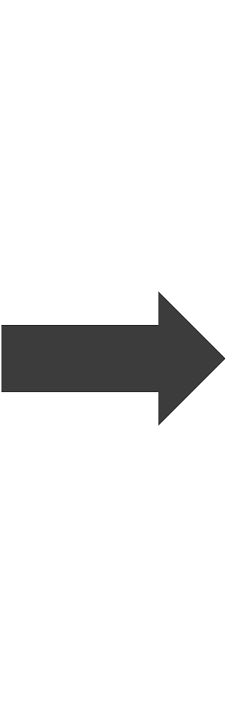}
\hspace{0.5cm}
\includegraphics[width=0.1\textwidth]{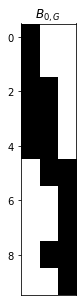}
\includegraphics[width=0.1\textwidth]{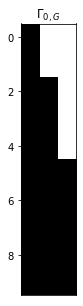}
\caption{An example of $\bB_0$, $\bGamma_0$, and $\bB_{0,G}$, $\bGamma_{0,G}$ after $\psi_G$ permutation.}
\end{figure}

\begin{theorem}
Let $(\bB_0,\bGamma_0)$ be a regular infinite loading pair with $\bGamma_0$ known, let $\bOmega_{0,n}$ be a $K \times n$ matrix with linearly independent rows, and let $\bSigma=diag(\sigma_1^2,\cdots)$ be a known infinite diagonal matrix such that $C_3 \leq \sigma_j^2 \leq C_4$ holds for $j=1,\cdots$, with constants $C_3, C_4>0$. The $j$-th row of $\infty \times n$ matrix $\bY$ is generated by $\mathcal{N}_n((\bB_0)_{j\cdot}\bOmega_{0,n},\sigma_j^2 \mathbf{I}_n)$. For every fixed $G$, let $\bY_{1:G}$ denote the matrix formed by permuting the first G rows of $\bY$ according to $\psi_G$ and consider modeling the $i$-th column of $\bY_{1:G}$ by $\mathcal{N}_G(\bB\bOmega_{\cdot i},\bSigma_G)$ for $i=1,\dots,n$ with $\bSigma_G=diag(\sigma_{\psi_G^{-1}(1)}^2,\cdots,\sigma_{\psi_G^{-1}(G)}^2)$. 
With a flat prior on each of $\bB$'s non-zero element according to the feature allocation matrix $\bGamma_{0,G}$ and a prior on $\bOmega$ that is invariant under right orthogonal transformations,  for a random draw $\bOmega$ from its posterior distribution, we have almost surely (with respect to the randomness in $\bY$) that \[||\bV(\bOmega_{0,n})_{1:k}^{\bot}\bV(\bOmega)_{1:k}^T||_F|\bY_{1:G}, \bSigma_G, \bGamma_{0,G}\rightarrow 0, \] 
for $k=1,\cdots,K$ as $G\rightarrow \infty.$

\end{theorem}

Theorem 5.3 is understood as the consistency (up to sign permutations) of $\bV(\bOmega)$ for fixed n and $G\asymp s \rightarrow \infty$, in the sense that $||\bV(\bOmega_{0,n})_{1:k}^{\bot}\bV(\bOmega)_{1:k}^T||_F$  converges to 0 for all $k$, which implies that the canonical angles between the row space of $\bV(\bOmega_{0,n})_{1:k}$ and that of $\bV(\bOmega)_{1:k}$ converge to 0 as $G\to \infty$. When these angles are all equal to 0, $\bV(\bOmega)$ differs from $\bV(\bOmega_{0,n})$ only by a sign for each row. Since the data provides no information on the signs, in the asymptotic regime with $G\asymp s \gg n \rightarrow \infty$, we can approximate  $\bV(\bOmega)$ drawn from its posterior distribution by a  random sign diagonal matrix $\bS$, i.e., a diagonal matrix with $i.i.d.$ random signs on the diagonal, times $\bV(\bOmega_{0,n})$. 

\subsubsection{Posterior sample consistency}\label{consistency}
Recall that from Section 5.1.2, for the basic Bayesian factor model with $G=s \gg n \rightarrow \infty$, $\bB\bK(\bOmega)/\sqrt{n}$ drawn from the posterior distribution can be asymptotically represented as the true loading matrix times a uniform random orthogonal matrix.  If the true feature allocation matrix is lower triangular, we have 
\begin{equation}\label{posterior_BK2}
\begin{split}
\bB^{(k)}\bK(\bOmega)_{1:k}/\sqrt{n}|\bY,\bOmega, \bSigma_G,\bGamma_{0,G} &\sim\bB_{0,G}^{(k)} (\bK(\bOmega_{0,n})_{1:k,1:k}/\sqrt{n}) \bV(\bOmega_{0,n})_{1:k}\bV(\bOmega)_{1:k}^T \\
&+
((\bY_{l_k:l_{k+1}-1}-\bB_{0,G}^{(k)}(\bOmega_{0,n})_{1:k})/\sqrt{n}) \bV(\bOmega)_{1:k}^T\\
&+\mathcal{N}_{(l_{k+1}-l_k)\times k}(\mathbf{0},\frac{1}{n}\mathbf{I}_k\otimes \bSigma_G^{(k)}),
\end{split}
\end{equation}
whose right hand side converges entry-wise in probability to $\bB_{0,G}^{(k)}\bS_{1:k,1:k}^T$ under the $G \asymp s \gg n \rightarrow \infty$ setting (by similar argument as in section 5.1.2). Note that $\bB^{(k)}\bK(\bOmega)_{1:k}=\bB_{l_k:l_{k+1}-1}\bK(\bOmega)$, we can therefore summarize the convergence of $\bB^{(k)}\bK(\bOmega)_{1:k}/\sqrt{n}$ to derive the convergence of posterior samples of $\bB \bK(\bOmega)/\sqrt{n}$ towards $\bB_{0,G} \bS^T$.

The posterior sample consistency (up to sign permutations) of the loading matrix is immediate once we have $\bK(\bOmega)/\sqrt{n}$, or equivalently $\bOmega \bOmega^T/n$, from its posterior distribution converging in probability to the identity matrix. The density in (\ref{posterior_K2}) indicates that the posterior distribution of $\bOmega \bOmega^T/n$ is contributed by two terms: the determinant  $\prod_{k=1}^K \bK(\bOmega)_{k,k}^{-(G-l_k+1)}$ and the model assumption represented by $p_{\bOmega}$. 
The determinant term creates singularities when $\bK(\bOmega)_{k,k}=0$ and the order of these ``poles'' $\sim s$.
When this term dominates, we observe the inflation phenomenon of posterior samples of the loading matrix. Meanwhile, the model assumption term can bound $\bK(\bOmega)$ away from these singularities by assigning little probability measure in their neighborhoods and also induces the convergence of $\bOmega \bOmega^T/n$ towards the identity matrix (through requirement (a) introduced in section 5.1). Consequently, the posterior behavior of $\bOmega \bOmega^T/n$ is influenced by both the increasing rate of $n,s$ and the choice of distribution $p_{\bOmega}$. Those $p_{\bOmega}$ that bounds away singularities with high probability and forces a fast convergence of $\bOmega \bOmega^T/n$ towards the identity matrix can allow a fast rate of $s$ going to infinity comparing to $n$, to guarantee the posterior consistency of the loading matrix. A simple and effective choice is to adopt the {\it $\sqrt{n}$-orthonormal factor model}. That is, we assume {\it a priori} that $\bOmega/\sqrt{n}$  is uniform in the Stiefel manifold $St(K,n)$. With this choice, we have $\bOmega \bOmega^T/n=I_K$ and that the posterior sample consistency of $\bB$ naturally holds even when $n$ has a rather slow growing rate compared with $s$. 

Our analysis regarding the relation between the factor assumption and the magnitude problem is specific to the independent spike and slab prior setup. But we believe that the magnitude problem, meaning that the column-wise magnitude of the loading matrix sampled from its posterior distribution is very sensitive to its prior distribution, exists for general priors in the ``Large s, Small n'' regime when the factors are only assumed to be normally distributed. When a more complicated prior is assigned on the loading matrix, the analysis becomes rather challenging and the magnitude problem may be expressed in other forms (as we will see in the simulations) rather than an `inflation' (inflation is typical for using non-informative priors on loading matrix). The intuition we gain from the analysis is that the factor assumption crucially impacts the strength of posterior contraction (through data) of the magnitude of the loading matrix towards its true value, and using $\sqrt{n}$-orthonormal factors achieves the strongest contraction thus allows more flexibility in prior assignment of the loading matrix while maintaining the posterior consistency.

\section{Numerical results}
\subsection{Modification of the Gibbs sampler}\label{modification}
In Section 5, we justify the adoption of the $\sqrt{n}$-orthonormal factor model in the ``Large s, Small n'' paradigm (i.e., the factor matrix $\bOmega$ scaled by $1/\sqrt{n}$ is uniform in the Stiefel manifold $St(K,n)$).
To construct a Gibbs sampler under this new 
factor model and the prior setup in Section~\ref{prior_setup} (denoted as SpSL-orthonormal factor model), we only need to revise the conditional sampling step of $\bOmega|\bY,\bB,\bSigma$ in the basic Gibbs sampler described in Section~\ref{standard-Gibbs}. 

Let $\bOmega_{k\cdot}$ denote the $k$-th row of the factor matrix and $\bOmega_{-k}$ denote the remaining  rows, all as column vectors. The conditional distribution $\bOmega_{k\cdot}|\bY,\bOmega_{-k},\bB,\bSigma$ is altered from a multivariate normal distribution to:
\begin{equation}\label{sample_Omega}
\pi(d\bOmega_{k\cdot}|\bY,\bOmega_{-k},\bB,\bSigma) \propto f(\bOmega_{k\cdot}; \ \bar{\bOmega}_{k\cdot},\bar{\sigma}^2 _k \mathbf{I}_n)\times p_{\bOmega_{-k}}(d\bOmega_{k\cdot})
\end{equation}
where $p_{\bOmega_{-k}}$ is the uniform measure on the centred $\sqrt{n}$-radius sphere in the orthogonal space of $\bOmega_{-k}$, and  $f(\bOmega_{k\cdot};\ \bar{\bOmega}_{k\cdot},\bar{\sigma}^2 _k \mathbf{I}_n)$ is the multivariate normal density function with mean $\bar{\bOmega}_{k\cdot}$ and covariance matrix $\bar{\sigma}^2_k \mathbf{I}_n$, with
$$\bar{\bOmega}_{k\cdot}=(\bB_{\cdot k}^T\bSigma^{-1}\bB_{\cdot k})^{-1}(\bY-\sum_{t\neq k} \bB_{\cdot t} \bOmega_{t \cdot}^T)^T\bSigma^{-1} \bB_{\cdot k},\  \  \ \bar{\sigma}^2_k =(\bB_{\cdot k}^T\bSigma^{-1}\bB_{\cdot k})^{-1}.$$

To sample from (\ref{sample_Omega}), we cut this $\sqrt{n}$-radius sphere by hyperplanes
that are orthogonal to vector $\bar{\bOmega}_{k\cdot}$ and denote this collection of intersections of the sphere and hyperplanes as $\{S_d\mid d \in (-\sqrt{n},\sqrt{n})\}$, where $d$ is the Euclidean distance between the origin and the hyperplane. Essentially, $\{S_d\}$ are ($n$-$k$)-dimensional spheres and every point in the same $S_d$ has the same multivariate normal density $f(\cdot;\bar{\bOmega}_{k\cdot},\bar{\sigma}^2 _k \mathbf{I}_n)$, so we can sample $\bOmega_{k\cdot}$ from (\ref{sample_Omega}) by first sampling $d$ from its marginal distribution and then uniformly sample  from sphere $S_{d}$ given the sampled $d$. Using the area formula of sphere, we can deduce the marginal distribution for $d$ as
\begin{equation}
\pi(d|\bY,\bOmega_{-k},\bB,\bSigma)\propto (n-d^2)^{(n-K-2)/2}\exp(\|\mathcal{P}_{\bOmega_{-k}^{\bot}}(\bar{\bOmega}_{k\cdot})\| d/\bar{\sigma}^2_k)
\end{equation}
and sample from this unimodal distribution using the Metropolis algorithm. The additional computational cost brought by the model revision only comes from the Metropolis algorithm and is almost negligible.  
\subsection{Comparison with alternative approaches}
We now revisit the synthetic example in Section~\ref{sec:synthetic}
to check the consistency 
of the posterior distribution of the loading matrix under the SpSL-orthonormal factor model and compare the MCMC performance with the sampler of two alternative approaches: a modified Ghosh-Dunson model (details provided in Appendix B) and the model from \cite{Bhattacharya2011Sparse} (applied with $\nu=3, a_{\sigma}=1,b_{\sigma}=0.3,a_1,a_2\sim \text{Gamma}(2,1)$). The factor dimensionality $K$ is fixed at 8 in all Gibbs samplers.

\begin{figure}[!htb]
\centering
\subfloat[The SpSL-orthonormal factor model]{\includegraphics[width=0.9\textwidth]{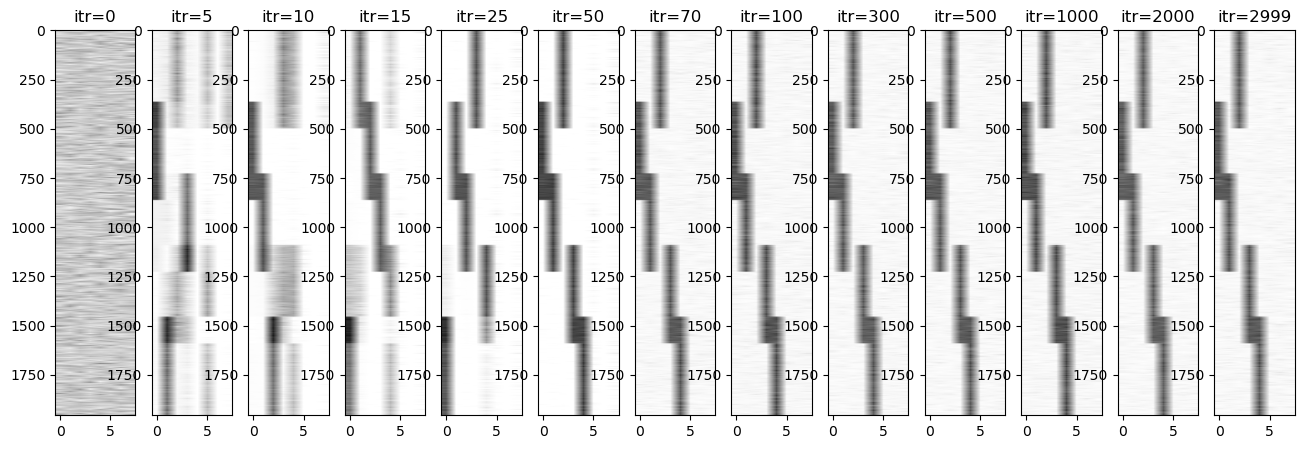}}

\subfloat[The modified Ghosh-Dunson model]{\includegraphics[width=0.9\textwidth]{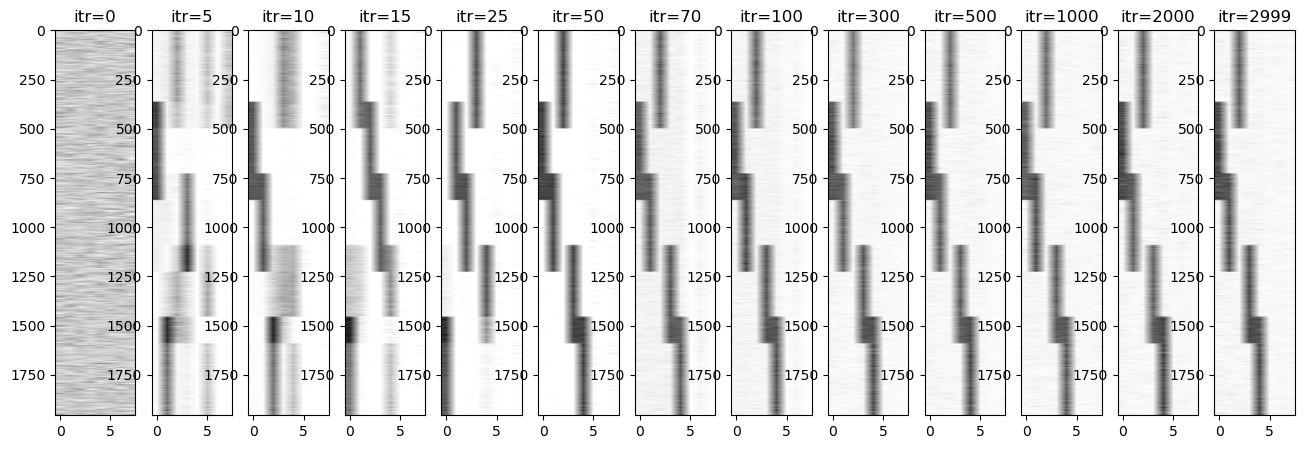}}

\subfloat[The model from \cite{Bhattacharya2011Sparse} ]{\includegraphics[width=0.9\textwidth]{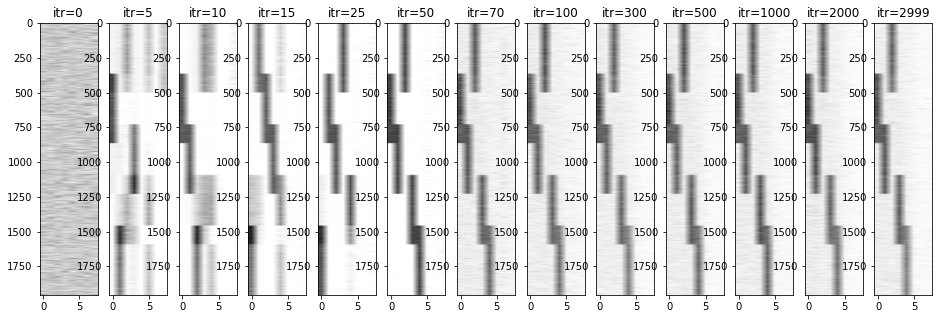}}
\caption{Heat-maps of $|\bB|$ in 3000 iterations of Gibbs sampler using specified models.}\label{example4}
\end{figure}

\begin{figure}[!htb]
     \centering
     \subfloat[The SpSL-orthonormal factor model]{\includegraphics[width=0.5\textwidth]{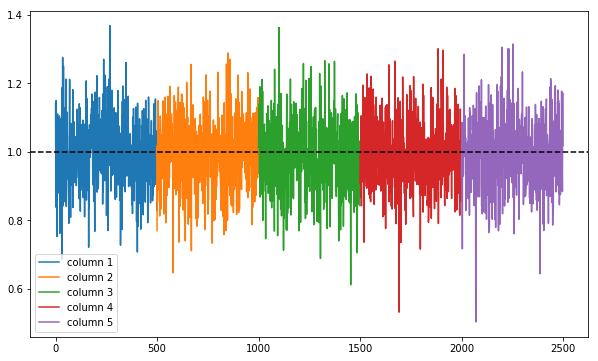}}
     \subfloat[The modified Ghosh-Dunson model]{\includegraphics[width=0.5\textwidth]{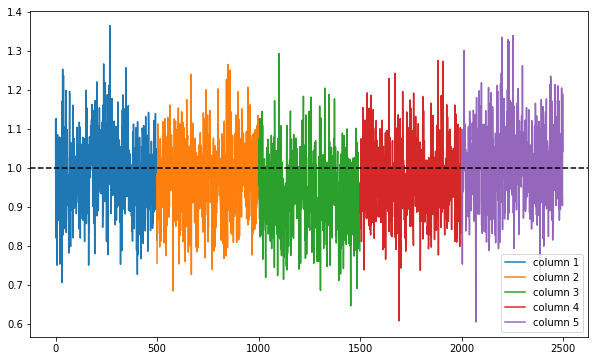}}
     
     \subfloat[The model from \cite{Bhattacharya2011Sparse}]{\includegraphics[width=0.5\textwidth]{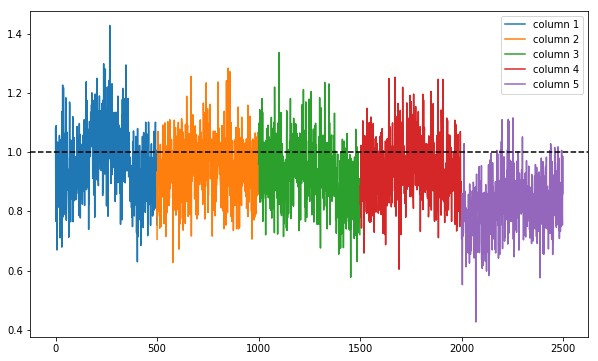}}
     \caption{\label{example5} Posterior means of the nonzero elements of loading matrix under three factor models. Nonzero elements are sorted first by the column index and then by the row index, both in ascending order, e.g. the first 500 entries colored in blue correspond to the posterior means of $\beta_{365,1},\cdots,\beta_{864,1}$.}
\end{figure}

Figure~\ref{example4} shows the heat map of $|\bB|$ in 3000 iterations. We perform the PXL-EM algorithm for the first 50 iterations and then Gibbs sampling in all three approaches, respectively, for the next 2950 iterations.
Figure~\ref{example5} shows the posterior means of the nonzero elements of loading matrix obtained by averaging over 2500 posterior samples after burn-in. For the SpSL-orthonormal factor model, the posterior means are nicely centered around the true value 1. For the other two approaches, there exhibit some `twists' in the column-wise magnitude of the posterior means. This is most obvious for the fifth (purple) column of panel (c) in Figure~\ref{example5}.

To understand the cause, we examine more closely the priors employed by the latter two approaches. These two approaches share the same idea of imposing dependency among the magnitudes of loading matrix elements within the same column via the decomposition  $\bB=\bQ\times\bD$ as we explained in the introduction. However, they differ in the scheme of learning factor dimensionality: the modified Ghosh-Dunson model adopts a SpSL prior on elements of $\bQ$, an Indian buffet process on $\bTheta$ and a diffuse prior on diagonals of $\bD$; whereas \cite{Bhattacharya2011Sparse} uses a continuous prior on elements of $\bQ$ and a shrinkage prior on $\bD$. The former model learns factor dimensionality though the shrinkage on the feature sparsity vector $\bTheta$ while the latter does so through the shrinkage on diagonals of $\bD$.
Under the ``Large s, Small n'' regime, using an informative prior on $\bD$ can be influential for the posterior of the column-wise magnitude and results in the `twists' in panel (c). As for the `twist' in panel (b), we think it is caused by the high auto-correlation among the samples generated by the Gibbs sampler for the modified Ghosh-Dunson model, so that the sample mean estimator still has a large Monte Carlo error using 2500 Gibbs samples. 

With a sufficient computation budget, Gibbs sampler for the modified Ghosh-Dunson model gives similar posterior results as our approach. This highlights another advantage of our approach---computational efficiency. When running the corresponding Gibbs sampler for 2500 rounds after burn-in, our approach with the SpSL-orthonormal factor model attained an average effective sample size (ESS) of 2758.8; whereas the ESS for the other two approaches are only 51.6 and 82.5, respectively, on average.
The computation times per iteration of Gibbs sampling for the three methods are 2.8, 2.2, and 1.8 seconds, respectively.  
Besides computational aspect, although both the SpSL-orthonormal factor model and the modified Ghosh-Dunson model give very similar numerical results after appropriately adjusting tuning parameters of the priors, our analysis rigorously justifies the consistency of the former model, whereas a similar theoretical study of the latter model is still beyond our reach.

\subsection{Robustness against prior specification}

Under the SpSL-orthonormal factor model, figure~\ref{example7} illustrates the posterior density of $\beta_{1,1}$ and $\beta_{1,3}$ (estimated by averaging over the conditional posterior densities) using slab priors with ranging variances. We tested with $\lambda_0=20,\ \lambda_1 \in \{0.001,0.01,0.1,0.5\}$ and the posterior distribution shows a great robustness against the choice of the slab prior. 

\begin{figure}[!htb]
     \centering
     \subfloat[$\beta_{1,1}$]{\includegraphics[width=0.41\textwidth]{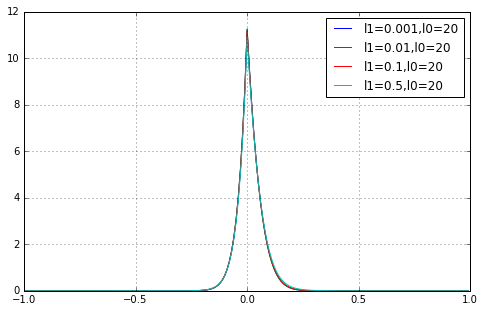}}
     \subfloat[$\beta_{1,3}$]{\includegraphics[width=0.41\textwidth]{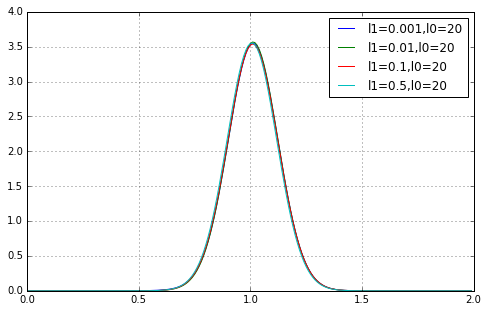}}
     \caption{\label{example7}Posterior densities of (a) $\beta_{1,1}$,  and (b) $\beta_{1,3}$, with $\lambda_1 \in \{0.001,0.01,0.1,0.5\}$, under SpSL-orthonormal factor model. The posterior densities are robust against the choice of slab priors.}
\end{figure}

At the end of Section~\ref{consistency}, we claim that restricting to $\sqrt{n}$-orthonormal factors grants more flexibility in prior assignment of the loading matrix while maintaining the posterior consistency. We  verify this claim by applying the prior setups from Ghosh-Dunson model and \cite{Bhattacharya2011Sparse} to the $\sqrt{n}$-orthonormal factor model. Note that we only need to revise the conditional sampling step $\bOmega|\bY,\bB,\bSigma$ in the Gibbs samplers as we did in Section~\ref{modification}.

\begin{figure}[!htb]
     \centering
     \subfloat[The modified Ghosh-Dunson model with normal factors]{\includegraphics[width=0.47\textwidth]{dbfm_post.png}}
     \hspace{0.5cm}
     \subfloat[The modified Ghosh-Dunson model with $\sqrt{n}$-orthonormal factors]{\includegraphics[width=0.47\textwidth]{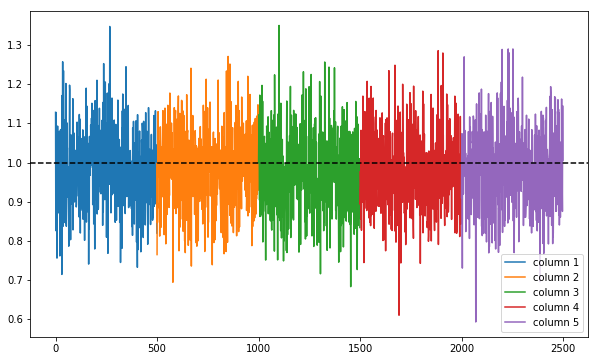}}
     
     \subfloat[The model of \cite{Bhattacharya2011Sparse} with normal factors]{\includegraphics[width=0.47\textwidth]{fm_post.png}}
     \hspace{0.5cm}
     \subfloat[The model of \cite{Bhattacharya2011Sparse} with $\sqrt{n}$-orthonormal factors]{\includegraphics[width=0.47\textwidth]{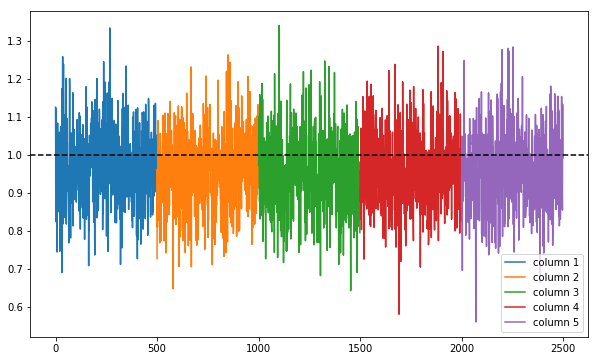}}
     \caption{\label{example8}Posterior mean of the nonzero elements of loading matrix using specified models. Column-wise magnitudes of the loading matrix from posterior are balanced after changing from normal factors (left panels) to $\sqrt{n}$-orthonormal factors (right panels).}
\end{figure}

Figure~\ref{example8} plots the posterior means of the nonzero elements of the loading matrix estimated by averaging over 2500 posterior samples. We observe that the `twists' in column-wise magnitudes disappear after switching to the $\sqrt{n}$-orthonormal factor model. 
Furthermore, the average ESS increases significantly from 51.6 and 82.5 to 2667.9 and 2296.6, respectively, for the two approaches since the source of high auto-correlations---strong tie between the magnitudes of the loading matrix and the factors, is removed by restricting the magnitude of factors to a specific value. Summary figures for credible intervals of the loading matrix elements of all implemented approaches are illustrated in Appendix D.2.

\section{Dynamic exploration with application}
Although the $\sqrt{n}$-orthonormal factor model can be coupled with general prior assignments on the loading matrix, we focus on the setup from \cite{rovckova2016fast} (i.e., the SpSL-orthonormal factor model), under which posterior consistency has a theoretical guarantee. 
When applying this framework to real data, 
the choice of the factor dimensionality $K$ as well as the penalty parameters $\lambda_0$ and $\lambda_1$ (parameters in the spike and the slab parts, respectively)  is crucial. 

The application of our Gibbs sampler  requires a successful implementation of the PXL-EM algorithm to search for a posterior mode that can serve to initialize the sampler. For the choice of $K$ when applying the sampler, we make two recommendations: (i) use the estimated number of factors from PXL-EM as a plug-in estimator for $K$; (ii) choose $K$ to be sufficiently large initially and discard the useless factors (whose corresponding $\{\gamma_{jk}\}_{j=1,\cdots,G}$ are all zero) in the sampling process, which is similar to the idea of choosing the number of factors adaptively from \cite{Bhattacharya2011Sparse}. More precisely, we discard  unless factors if there are any, and append a null factor whenever there is no useless factor remained. Though this adaptive approach also provides posterior samples for $K$, it is worth mentioning that the computational complexity of the Gibbs sampler scales linearly with the factor dimensionality $K$.

The penalty parameters determine the threshold for a loading matrix's element to follow either a spike or a slab prior. For the PXL-EM algorithm, Ro{\v{c}}kov{\'a} and George proposed a dynamic posterior exploration process to help searching for the MAP in a sequence of prior settings as well as determining the appropriate value for these penalty parameters. Initially, they fix $\lambda_1$ at a small value and gradually increase $\lambda_0$ until the solution path is stabilized. The solution given by the PXL-EM under the final value of $\lambda_0$ approximates the MAP estimate under a flat and point mass mixture prior on loading matrix elements and is proposed as the estimator for parameters. 
The same procedure can be applied to the full posterior inference based on the SpSL-orthonormal factor model.

We observed a similar stabilization of the posterior distributions of every nonzero loading element when performing dynamic exploration for the SpSL-orthonormal factor model, which is  illustrated in the application of our method  to the cerebrum microarray data from AGEMAP (Atlas of Gene Expression in Mouse Aging Project) database of \cite{zahn2007agemap}, which was analyzed by \cite{rovckova2016fast} using their PXL-EM algorithm. For every mice individual in this dataset  (5 males and 5 females, at four age periods), cerebrum microarray expression data from 8932 genes are recorded, observations $\by_i, i=1,\cdots,40$ for the factor model are taken to be the residuals of the expression values for  each of the 8932 genes regressed on age and gender with an intercept.

We ran the posterior sampler initialized at the MAP detected by the PXL-EM algorithm with $\lambda_1=0.001, \alpha=1/G$, and $\lambda_0$ gradually increasing in the sequence of 12,15,20,30,40. As the detected factor dimensionality by the PXL-EM algorithm is 1, we specify $K$ to be 1 in our framework. Figure~\ref{example9} demonstrates the evolution of the posterior density of $\beta_{2873,1}$ and $\beta_{1,1}$ as $\lambda_0$ changes.

\begin{figure}[!htb]
\centering
\subfloat[$\beta_{2873,1}$]{\includegraphics[width=0.45\textwidth]{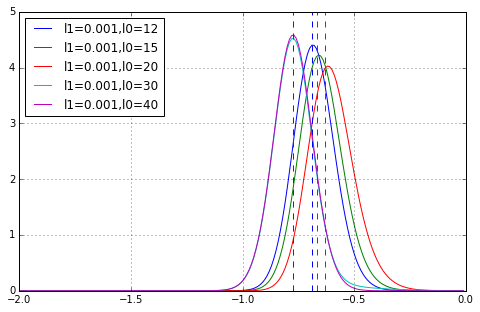}}
\subfloat[$\beta_{1,1}$]{\includegraphics[width=0.45\textwidth]{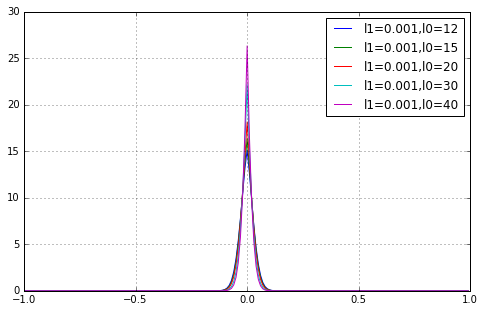}}
\caption{Posterior pdf of (a) $\beta_{2873,1}$ and (b) $\beta_{1,1}$ under  SpSL-orthonormal factor model with increasing $\lambda_0$ }\label{example9}
\end{figure}

The posterior distribution of $\beta_{1,1}$  centers at 0 and becomes more and more spiky as $\lambda_0$ increases. For the nonzero element $\beta_{2873,1}$, its posterior distributions resemble the normal distribution with a relative stable variance.
The posterior mean of $\beta_{2873,1}$ first moves towards zero and then away and stabilizes. This change of direction is caused by the alteration of its slab indicator $\gamma_{2873,1}$ from 0 to 1 in posterior samples, 
in which case the posterior distribution of $\beta_{j,k}$ is only influenced by the slab, but not the spike prior. Vertical dotted lines are the MAP estimates, which are close to the posterior means.
Having recognized that the stabilization of the MAP estimates and the posterior distributions occur almost simultaneously as $\lambda_0$ increases, in practice we can find the ideal pair of penalty parameters such that the posterior distribution is stabilized by looking for the stabilization of the MAP estimates instead of sampling from the posterior distribution with $\lambda_0$ on multiple levels. More summary and comparative figures of the posterior simulation are illustrated in Appendix D.1 with $\lambda_0=30$.


\begin{figure}[!htb]
\centering
\subfloat[The SpSL-orthonormal factor model]{\includegraphics[width=1\textwidth]{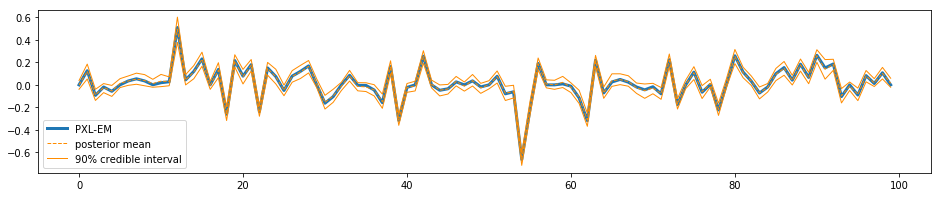}}

\subfloat[The modified Ghosh-Dunson model]{\includegraphics[width=1\textwidth]{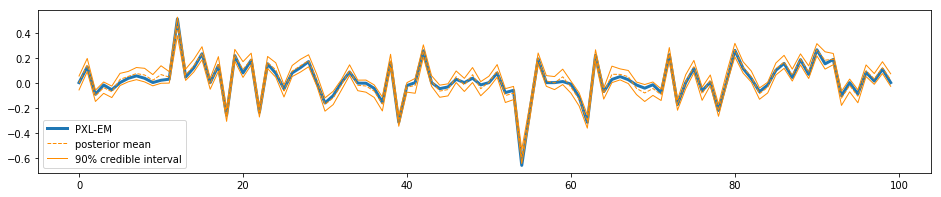}}
\caption{Posterior mean and credible interval of $\beta_{1,1},\cdots,\beta_{100,1}$  estimated from samples of specified model and the MAP estimate from PXL-EM algorithm}\label{example10}
\end{figure}

Figure~\ref{example10} provides a comparison between the posterior inference results from the SpSL-orthonormal factor model ($\lambda_0=30$) and the modified Ghosh-Dunson model ($\lambda=0.001$, $\lambda_0=200$)
which shows that the two models give very similar posterior  credible intervals (computed using 1000 posterior samples after burn-in) for the loading matrix, and both posterior means are also very close to the MAP estimator from PXL-EM algorithm. 
Additionally, the Gibbs sampler for the SpSL-orthonormal factor model results in a much larger ESS compared to that for the Ghosh-Dunson model (e.g.,  the ESS for $\beta_{55,1}$ are 905.0 and 42.7 for the two methods, respectively). We omit 
scientific interpretations of the inference results
since our goal is only to verify that our procedure gives similar results as those in \cite{rovckova2016fast}  based on points estimates of the normal factor model, and to show how to conduct the full Bayesian analysis properly and efficiently for this dataset.

In summary, we can start our  Bayesian inference  for the SpSL-orthonormal factor model by first choosing a small $\lambda_1$ and a sequence of increasing $\lambda_0$, denoted as $\{\lambda_0^{(t)}\}_{t=1,\cdots}$. We then run the PXL-EM algorithm sequentially with $\lambda_1$ and $\lambda_0^{(t)}$ for $t=1,\cdots$, with parameters  initialized at the MAP estimate found in the previous round.  
The process is terminated when the difference between the new MAP estimate and the one from the previous round is below a chosen threshold. Afterward, we run our Gibbs sampler under the SpSL-orthonormal factor model using the final pair of penalty parameters with $\bB,\bSigma,\bTheta$ and $K$ initialized at the MAP estimate and $\bOmega,\bGamma$ initialized with random draws from their  domains.

\section{Discussion}
A primary intention of our work is to provide an efficient posterior sampler for the Bayesian factor model in high dimensions and show its consistency. \cite{rovckova2016fast}'s sparse Bayesian factor model framework serves as a 
promising starting point, for both its explicit encoding of the sparsity 
and its providing of a fast posterior mode finding algorithm.
By analyzing the magnitude inflation problem of the posterior samples of the loading matrix under the prior setup of \cite{rovckova2016fast}, we  propose the $\sqrt{n}$-orthonormal factor model as a practical remedy, which not only processes posterior consistency and robustness against prior settings, but also dramatically improves the computational efficiency. Our work naturally bridges the gap between the point estimation based on posterior modes  and the full Bayesian analysis under the SpSL factor modeling framework.
Besides our proposed solution, i.e., enforcing a common scale and orthogonality among the factors, \cite{bernardo2003bayesian} and \cite{ghosh2009default} provided another perspective, which is  to reduce the dimensionality of diffuse parameters in the prior  to ensure that they do not overwhelm  the data. Their approach allows the factors to have different variances, 
but restricts elements of  the loading matrix to follow standard Gaussian {\it a priori}. In this article, we provide a further modification of their model by imposing a SpSL prior on the loading matrix's elements, which allows a greater flexibility in handling sparsity in high dimensions (details in Appendix B).

Using the prior from \cite{rovckova2016fast}, we are able to show theoretically that the adoption of a strict $\sqrt{n}$-orthonormal factor assumption can ensure posterior consistency.
But this type of rigorous analysis for other models, including the Ghosh-Dunson model and its modification, still evades our vigorous attempts. Furthermore, in some follow up work, informative priors were assigned to the diffuse parameters in Ghosh-Dunson model, and it is unclear how these priors influence the posterior magnitude of the loading matrix generally. 
Interests for future exploration may be focused on the design of dependent priors for easy posterior sampling as well as the justification of posterior consistency when using such priors. 
The $\sqrt{n}$-orthonormal factor model itself is also interesting, since the posterior consistency under this model is empirically more robust against prior specification of the loading matrix in the high dimensional setting. It would be interesting to see a mathematical formulation of this empirical result in future works.

\section*{Acknowledgement} 
This research is supported in part by the  National Science Foundation of USA  Grants DMS-1613035, DMS-1712714 and DMS-1903139. 

\nocite{Bhattacharya2011Sparse,Liu2000Generalised,Liu1999Parameter,rovckova2016fast,fruehwirth2018sparse,fruehwirth2018sparse,geman1984stochastic,bernardo2003bayesian,zahn2007agemap,ghosh2009default,natarajan1998gibbs,kass1996selection,efron1973discussion,tanner1987calculation,gelfand1990sampling,liu2008monte,pati2014posterior,xie2018bayesian,kaiser1958varimax,carvalho2008high}
\bibliographystyle{chicago}
\bibliography{reference.bib}

\newpage
\begin{appendix}

\centerline{\huge\bf Appendix}

\section{Scaling group moves}
To see how the posterior distribution of the loading matrix is influenced by the SpSL prior, we need to observe the sample behavior at equilibrium with different priors. Due to the strong ties between the  loading matrix and the latent factors, samples are inflating slowly along the basic Gibbs sampling iterations, which demonstrates the slow mixing behavior of the Gibbs sampler.

A promising way to improve Markov Chain Monte Carlo (MCMC) convergence is to add a group move into the sampler.  \cite{Liu2000Generalised}  proposed ``generalized Gibbs sampling", which can be seen as a generalization of \cite{Liu1999Parameter}  for conditional sampling along the trajectories of any designed transformation group.  By taking advantage of the model structure and proposing a group trajectory that can cross various significant local modes, this group move can dramatically improve the MCMC convergence. The following theorem from \cite{Liu2000Generalised} characterizes how a group move should be conducted.

\begin{theorem}{\textbf{(Liu and Sabatti(2000))}}
Let $\pi$ be an arbitrary distribution on a space $\mathscr{Z}$, and suppose $t_{\alpha}(z): \mathscr{Z}\rightarrow \mathscr{Z}$ is a transformation parameterized by $\alpha \in \mathscr{A}$. Assume there is a group structure on both $\mathscr{A}$ and the transformation family, and a left-Haar measure $H$ on $\mathscr{A}$. If $z$ follows distribution $\pi$ and $\alpha$ is drawn from
\begin{equation}
\pi(\alpha|z) \propto \pi(t_{\alpha}(z))\left|\frac{\partial t_{\alpha}(z)}{\partial z}\right|H(d\alpha), \label{g-move}
\end{equation}
then $t_{\alpha}(z)$ follows distribution $\pi$.
\end{theorem}

If $\pi$ in Theorem A.1. is the full posterior distribution, then $t_{\alpha}$ generated by the conditional distribution (\ref{g-move})  gives a transformation that preserves the target distribution $\pi$.  We can add this transformation after each round of Gibbs sampling to improve convergence. To design group moves that can move the loading matrix and factors jointly in the synthetic example, we
consider the following group of scale transformations for $k=1,\cdots,K$:
$$t_{\alpha_k}(\beta_{1k},\cdots,\beta_{Gk},\omega_{1k},\cdots,\omega_{nk}) = \left(\alpha_k\beta_{1k},\cdots,\alpha_k\beta_{Gk},\frac{1}{\alpha_k}\omega_{1k},\cdots,\frac{1}{\alpha_k}\omega_{nk} \right),$$
and draw $\alpha_k$ sequentially from:
$$p(d \alpha_k) \propto \prod_{j=1}^G ((1-\gamma_{jk})\psi(\alpha_k\beta_{jk}|\lambda_{0})+\gamma_{jk}\psi(\alpha_k\beta_{jk}|\lambda_1))\times \prod_{i=1}^n \exp(-\frac{\omega_{ik}^2}{2\alpha_k^2})\times\alpha_k^{G-n-1}\ d\alpha_k$$

We design these group moves to rescale each column since we observe a synchronous inflation within every column during Gibbs sampling and changes of magnitude are encumbered due to the strong connection between factors and loading. These scaling group moves are cheap to implement since the conditional distribution of  $\alpha_k$ is  a univariate and unimodal distribution. More delicate moves such as linear restructuring (corresponding to `rotate' the loading in PXL-EM) $t_A(B,\Omega): B,\Omega \rightarrow BA, A^{-1}\Omega$ can be  difficult to implement in practice.

\section{The modified Ghosh-Dunson model}
Since the magnitude inflation is associated with the overdose of independent slab priors on the loading matrix, an immediate counter measure would be to   control the number of slab priors used. \cite{ghosh2009default}  proposed to use an inverse gamma prior for the variance of the normal factors and impose the standard Gaussian prior  on elements of the loading matrix, which will be called th Ghosh-Dunson model.
Here we propose a modified Ghosh-Dunson model  by relocating the variance parameters of the factors to the loading matrix and imposing a  SpSL prior on its elements:
  \begin{equation}\label{RGD}
   \begin{split}
  \textbf{Model:}\ 
  &\by_i\mid \bomega_i,\bB,\bSigma\stackrel{i.i.d.}{\sim} \mathcal{N}_G(\bB\bomega_i,\bSigma),\ \bomega_i \stackrel{i.i.d.}{\sim} \mathcal{N}_K(\mathbf{0},\mathbf{I}_K)\\
  \textbf{Priors:}\ &\beta_{jk}=q_{jk}r_k,\ \ p(r_k|\lambda)= \psi(r_k|\lambda);\\
  &p(q_{jk}|\gamma_{jk},\lambda_{0},\lambda_1)=(1-\gamma_{jk})\psi(q_{jk}|\lambda_{0})+\gamma_{jk}\psi(q_{jk}|\lambda_1), \ \ \lambda_{0}\gg \lambda_1; \\
&\gamma_{jk}|\theta_{k} \sim  \text{Bernoulli}(\theta_{k}) \ \mbox{independently; } \\
\ &\theta_{k}=\prod^k_{l=1} \nu_l,\  \ \nu_l \stackrel{i.i.d.}{\sim}  \text{Beta}(\alpha,1);\\
&\sigma_j^2\stackrel{i.i.d.}{\sim}  \text{Inverse-Gamma}(\eta/2,\eta \varepsilon/2).
\end{split}
\end{equation}
where $\beta_{jk}$ denote the $(j,k)^{th}$ element of $\bB$ and $\psi(\cdot |\lambda)$ is the normal density with precision $\lambda$. We chose $\lambda_0$ large and $\lambda_1=1$.
 
In this framework, each loading element $\beta_{jk}$ is expressed as the product of a column-wise magnitude parameter $r_k$ and the `normalized' loading $q_{jk}$.
\cite{ghosh2009default}'s original model corresponds to assuming $\gamma_{jk}\equiv \theta_k\equiv 1$, i.e., a normal instead of mixture normal prior for the $q_{jk}$.
We impose a diffuse normal prior on the $r_k$'s and a  SpSL prior on $q_{jk}$. With this dependent prior specification, the number of the ``slab parameters" is greatly reduced (all elements in each column of $\bB$ share a common ``slab parameter'' $r_k$), while marginally the prior on each $\beta_{jk}$ is the same as that of the independent SpSL prior. This prior setup on the loading matrix is similar to the one in the hierarchical linear model in \cite{jia2007mapping} where $\bOmega$ is prescribed, and the prior setup on $\bB$ establishes connections between rows of the loading matrix to prevent the degeneration of the original model to multiple independent linear regressions. However, the hierarchical linear model is not subject to the inflation problem even if completely independent priors are imposed on the loading matrix since $\bOmega$ is already prescribed.

 Although the dependent slab prior specification is an effective way for resolving the posterior inflation problem, the justification of the posterior consistency is rather difficult under this framework. We simply provide some numerical results in Section 6 and 7 to compare  the posterior distribution based on the modified Ghosh-Dunson model (\ref{RGD}) with that   resulting from  our strategy of imposing the $\sqrt{n}$-orthonormal factor assumption. The simulations are performed with $\alpha=1/G,\eta=\epsilon=1,\lambda=0.001,\lambda_0=200, \lambda_1=1$ and $K=8$ (in Section 6) / $1$ (in Section 7) using a Gibbs sampler starting from the MAP identified by the PXL-EM algorithm.
 
\section{Mathematical Proofs}
\subsection{Proof of Theorem 4.1}
\begin{proof}
Let $\bbeta_1$ be the vector formed by the $\beta_{jk}$'s with their corresponding $\gamma_{jk}=1$ and let $\bbeta_0$ be the vector formed by $\beta_{jk}$'s with their corresponding $\gamma_{jk}=0$.
\begin{equation}
\begin{split}
\pi(\bB|\bY,\bSigma,\bGamma,m) 
\propto & f_m(\bbeta_1,\bbeta_0) \equiv \prod_{\{j,k:\gamma_{jk}=1\} }\phi_m(\beta_{jk}) \prod_{\{j,k:\gamma_{jk}=0\} }\psi(\beta_{jk})\\
\times &|\bB\bB^T+\bSigma|^{-{\frac{n}{2}}} \exp \left\{ -\frac{1}{2}\text{tr} \left[(\bB\bB^T+\bSigma)^{-1}(\sum_{i=1}^n \by_i \by_i^T)\right] \right\}
\end{split}
\end{equation}

Let $\lambda_1(M)\geq \cdots\geq\lambda_G(M)$ denote the eigenvalues of a matrix $M$ 
and let $\mu_1\geq\cdots\geq \mu_G$ be the eigenvalues of $\bB\bB^T+\bSigma$. According to Weyl's inequality, $$\lambda_j(\bB\bB^T)+\lambda_1(\bSigma)\geq\mu_j\geq \lambda_j(\bSigma),  \ \ j=1,\cdots,G,$$
we have 
\begin{equation}
\begin{split}
\sum_{j=1}^G \frac{\lambda_j(\bY\bY^T)}{\lambda_j(\bB\bB^T)+\lambda_1(\bSigma)} \leq\sum_{j=1}^G \frac{\lambda_j(\bY\bY^T)}{\mu_{j}}\leq \text{tr}\left[(\bB\bB^T+\bSigma)^{-1}(\sum_{i=1}^n \by_i \by_i^T)\right]\\
\leq \sum_{j=1}^G \frac{\lambda_j(\bY\bY^T)}{\mu_{G+1-j}} \leq \sum_{j=1}^G \frac{\lambda_j(\bY\bY^T)}{\lambda_{G+1-j}(\bSigma)}
\end{split}
\end{equation}
Note that $\lambda_j(\bB)=0 $ for $ j>K$, so we have:
\begin{equation}\label{23}
\sum_{j=K+1}^G \frac{\lambda_j(\bY\bY^T)}{\lambda_1(\bSigma)} \leq \text{tr}\left[(\bB\bB^T+\bSigma)^{-1}(\sum_{i=1}^n \by_i \by_i^T)\right] \leq \sum_{j=1}^G \frac{\lambda_j(\bY\bY^T)}{\lambda_{G+1-j}(\bSigma)}
\end{equation}
According to the Minkowski determinant theorem, $|\bB\bB^T+\bSigma|\geq |\bSigma|$. Furthermore, \[|\bB\bB^T+\bSigma|=\prod_{j=1}^G \mu_j \leq \prod_{j=1}^G (\lambda_j(\bB\bB^T)+\lambda_1(\bSigma)) \leq (\lambda_1(\bSigma))^{G-K} \prod_{j=1}^K(||\bB\bB^T||_F+\lambda_1(\bSigma)).\] 
Combining this with (\ref{23}), we have
\begin{equation}
|\bB\bB^T+\bSigma|^{-n/2}\exp\left\{-\frac{1}{2}\text{tr}\left[(\bB\bB^T+\bSigma)^{-1}(\sum_{i=1}^n \by_i \by_i^T)\right] \right\}\leq |\bSigma|^{-n/2} \text{exp}(-\frac{1}{2} \sum_{j=K+1}^G \frac{\lambda_j(\bY\bY^T)}{\lambda_1(\bSigma)})
\end{equation}
and
\begin{equation}
\begin{split}
&|\bB\bB^T+\bSigma|^{-n/2}\exp\left\{-\frac{1}{2}\text{tr}\left[(\bB\bB^T+\bSigma)^{-1}(\sum_{i=1}^n \by_i \by_i^T)\right] \right\}\\
\geq &(\lambda_1(\bSigma))^{-n(G-K)/2} \prod_{j=1}^K(||\bB||_F^2+\lambda_1(\bSigma))^{-n/2} \ \text{exp}(-\frac{1}{2} \sum_{j=1}^G \frac{\lambda_j(\bY\bY^T)}{\lambda_{G+1-j}(\bSigma)}).
\end{split}
\end{equation}
Therefore,
\begin{equation}
\begin{split}
&\int_{\bB \in S} f_m(\bbeta_1,\bbeta_0) \ d \bB\\
\leq &\int_{\bB \in S} d\bB \ |\bSigma|^{-n/2} \text{exp}(-\frac{1}{2} \sum_{j=K+1}^G \frac{\lambda_j(\bY\bY^T)}{\lambda_1(\bSigma)}) \ (\text{max}_{\beta}(\phi_m(\beta)))^{\#\{\gamma_{jk}=1\} }\ (\text{max}_{\beta}(\psi(\beta)))^{\#\{\gamma_{jk}=0\} }\\
=& C_1(\text{max}_{\beta}(\phi_m(\beta)))^{\#\{\gamma_{jk}=1\} }
\end{split}
\end{equation}
For a constant $R>0$,
\[
\begin{split}
&\int_{|\bbeta_0|\leq R} \int_{\bbeta_1\in S_m^{\#\{\gamma_{jk}=1\}}} f_m(\bbeta_1,\bbeta_0) \ d \bbeta_1 d\bbeta_0\\
\geq &\int_{|\bbeta_0|\leq R} \int_{\bbeta_1\in S_m^{\#\{\gamma_{jk}=1\}}} \prod_{j=1}^K(||\bB||_F^2+\lambda_1(\bSigma))^{-n/2} \  \ d \bbeta_1 d\bbeta_0\\
\times&(\lambda_1(\bSigma))^{-n(G-K)/2} \text{exp}(-\frac{1}{2} \sum_{j=1}^G \frac{\lambda_j(\bY\bY^T)}{\lambda_{G+1-j}(\bSigma)}) (C \ \text{max}_{\beta}(\phi_m(\beta)))^{\#\{\gamma_{jk}=1\} } (\text{min}_{\beta<R}(\psi(\beta)))^{\#\{\gamma_{jk}=0\} }
\end{split}
\]
\begin{equation}
\begin{split}
\geq &C_2 (\text{max}_{\beta}(\phi_m(\beta)))^{\#\{\gamma_{jk}=1\} } \int_{|\bbeta_0|\leq R} \int_{\bbeta_1\in S_m^{\#\{\gamma_{jk}=1\}}} \prod_{j=1}^K(||\bB||_F^2+\lambda_1(\bSigma))^{-n/2} \ d \bbeta_1 d\bbeta_0\\
\rightarrow & C_2 (\text{max}_{\beta}(\phi_m(\beta)))^{\#\{\gamma_{jk}=1\} } \int_{|\bbeta_0|\leq R} \int_{\bbeta_1\in \mathcal{R}^{\#\{\gamma_{jk}=1\}}} \prod_{j=1}^K(||\bB||_F^2+\lambda_1(\bSigma))^{-n/2} \ d \bbeta_1 d\bbeta_0
\end{split}
\end{equation}
as $m\rightarrow \infty$ following the monotone convergence theorem. We also know that
\begin{equation}
\begin{split}
&\int_{|\bbeta_0|\leq R} \int_{\bbeta_1\in \mathcal{R}^{\#\{\gamma_{jk}=1\}}} \prod_{j=1}^K(||\bB||_F^2+\lambda_1(\bSigma))^{-n/2} \ d \bbeta_1 d\bbeta_0\\
\geq &\int_{|\bbeta_0|\leq R} \int_{\bbeta_1\in \mathcal{R}^{\#\{\gamma_{jk}=1\}}} \prod_{j=1}^K(\|\bbeta_1\|^2+R^2+\lambda_1(\bSigma))^{-n/2} \ d \bbeta_1 d\bbeta_0\\
=& (\int_{|\bbeta_0|\leq R} d\bbeta_0) \int_{\bbeta_1\in \mathcal{R}^{\#\{\gamma_{jk}=1\}}} \prod_{j=1}^K(\|\bbeta_1\|^2+R^2+\lambda_1(\bSigma))^{-n/2} \|\bbeta_1\|^{\#\{\gamma_{jk}=1\}-1}\ d \|\bbeta_1\| d(\gamma(\bbeta_1)) 
\end{split}
\end{equation}
from the polar coordinate transformation, of which the last term goes to infinity since $\#\{\gamma_{jk}=1\}\geq n\times K$. Taken together, we have shown that 
\begin{equation}
\text{lim}_{m\rightarrow \infty} \frac{\int_{\bB \in S} f_m(\bbeta_1,\bbeta_0) \ d \bB}{\int_{|\bbeta_0|\leq R} \int_{\bbeta_1\in S_m^{\#\{\gamma_{jk}=1\}}} f_m(\bbeta_1,\bbeta_0) \ d \bbeta_1 d\bbeta_0}=0,
\end{equation}
which implies the theorem.
\end{proof}
\subsection{Proof of Theorem 4.2}
\begin{proof}
By marginalizing out $\bOmega$ from the full posterior distribution, we know that:
\begin{equation}
\pi(\bB|\bY,\bSigma,\bGamma,m) \propto \int f(\bY|\bB,\bOmega,\bSigma) f(\bOmega) d\bOmega \prod_{\{jk:\gamma_{jk}=1\} }\frac{\phi_m(\beta_{jk})}{\phi_m(0)} \prod_{\{jk:\gamma_{jk}=0\} }\psi(\beta_{jk})=\pi_m^u(\bB)
\end{equation}
\begin{equation}
\pi(\bB|\bY,\bSigma,\bGamma,\infty) \propto \int f(\bY|\bB,\bOmega,\bSigma) f(\bOmega) d\bOmega  \prod_{\{jk:\gamma_{jk}=0\} }\psi(\beta_{jk})=\pi_{\infty}^u(\bB)
\end{equation}
For any Borel set $S$, $\int_S \pi_m^u(\bB) d\bB \leq \int_S \pi_{\infty}^u(\bB) d\bB <\infty $, by the dominant convergence theorem we have:
\begin{equation}
\lim_{m\rightarrow \infty}\int_S \pi_m^u(\bB) d\bB = \int_S \pi_{\infty}^u(\bB) d\bB,
\end{equation}
\begin{equation}
\lim_{m\rightarrow \infty} \left. \int_S \pi_m^u(\bB) d\bB \right/\int_{\mathcal{R}^{G\times K}} \pi_m^u(\bB) d\bB = \left.\int_S \pi_{\infty}^u(\bB) d\bB \right/ \int_{\mathcal{R}^{G\times K}} \pi_{\infty}^u(\bB) d\bB.
\end{equation}
This means that $\bB|\bY,\bSigma,\bGamma,m$ converges to $\bB|\bY,\bSigma,\bGamma,\infty$  in distribution as $m \rightarrow \infty$.
\end{proof}

\subsection{Proof of Lemma 5.1}
 \begin{proof}
 For $\epsilon >0 \text{ and } L>0$,
\begin{equation}\label{34}
\begin{split}
P(||\bV(\bOmega_{0,n})^{\bot}\bV(\bOmega)^T||_F>\epsilon|\bY, \bSigma_G)
\leq 1/(1+\frac{P(||\bV(\bOmega_{0,n})^{\bot}\bV(\bOmega)^T||_F<\epsilon/L|\bY, \bSigma_G)}{P(||\bV(\bOmega_{0,n})^{\bot}\bV(\bOmega)^T||_F>\epsilon|\bY, \bSigma_G)}). 
\end{split}
\end{equation}
From
\begin{equation}
\pi(d \bV(\bOmega)|\bY,\bSigma)\propto \text{exp}\Big(\sum_{j=1}^G\frac{1}{2\sigma_j^2}\|\mathcal{P}_{\bV(\bOmega)}(\bY_{j\centerdot})\|^2\Big)m(d \bV(\bOmega)),
\end{equation}
we can compute
\begin{equation}\label{35}
\begin{split}
&P(||\bV(\bOmega_{0,n})^{\bot}\bV(\bOmega)^T||_F<\epsilon/L|\bY, \bSigma_G) \\
=&C \int_{\{\bV(\bOmega):||\bV(\bOmega_{0,n})^{\bot}\bV(\bOmega)^T||_F<\epsilon/L\}} \text{exp}\Big(-\sum_{j=1}^G\frac{1}{2\sigma_j^2}\|\mathcal{P}_{\bV(\bOmega)^{\bot}}(\bY_{j\centerdot})\|^2\Big)m(d \bV(\bOmega))\\
=&C \int_{\{\bV(\bOmega):||\bV(\bOmega_{0,n})^{\bot}\bV(\bOmega)^T||_F<\epsilon/L\}} \text{exp}\Big(-\frac{1}{2}||\bV(\bOmega)^{\bot}\bV(\bOmega_{0,n})^T \bK(\bOmega_{0,n})^T \bB_{0,G}^T\bSigma_G^{-1/2}||_F^2\Big)m(d \bV(\bOmega))\\
\geq & C \int_{\{\bV(\bOmega):||\bV(\bOmega_{0,n})^{\bot}\bV(\bOmega)^T||_F<\epsilon/L\}} \text{exp}\Big(-\frac{1}{2}||\bV(\bOmega)^{\bot}\bV(\bOmega_{0,n})^T||_F^2 \lambda_{max}(\bSigma_G^{-1/2}\bB_{0,G}\bK(\bOmega_{0,n}))^2\Big)m(d \bV(\bOmega))\\
\geq & C m_n(\{\bV: ||\bV_0 \bV^T||_F < \frac{\epsilon}{L}\})\times  \exp\Big(-\frac{1}{2}\frac{\epsilon^2}{L^2}\lambda_{max}(\bSigma_G^{-1/2}\bB_{0,G}\bK(\bOmega_{0,n}))^2\Big),
\end{split}
\end{equation}
where $\bV_0$ is a $K \times n$ orthonormal matrix. Similarly, we can derive
\begin{equation}\label{36}
\begin{split}
&P(||\bV(\bOmega_{0,n})^{\bot}\bV(\bOmega)^T||_F>\epsilon|\bY, \bSigma_G) \\
\leq & C m_n(\{\bV: ||\bV_0 \bV^T||_F >\epsilon\}) \exp\Big(-\frac{1}{2}\epsilon^2\lambda_{min}(\bSigma_G^{-1/2}\bB_{0,G}\bK(\bOmega_{0,n}))^2\Big).
\end{split}
\end{equation}
Inserting (\ref{35}) and (\ref{36}) to (\ref{34}), we complete the proof.
\end{proof}

\subsection{Proof of Theorem 5.2}

\begin{proof}
First, we show a strong uniform law of large number that:
\begin{equation}\label{37}
\lim_{G \rightarrow \infty}\sup_{\bOmega}\Big|\frac{1}{G}\sum_{j=1}^G\frac{1}{2\sigma_j^2}\|\mathcal{P}_{\bV(\bOmega)}(\bY_{j\centerdot})\|^2-\frac{1}{G}\sum_{j=1}^G\frac{1}{2\sigma_j^2}\mathbb{E}\|\mathcal{P}_{\bV(\bOmega)}(\bY_{j\centerdot})\|^2 \Big|=0 \ \ a.s.
\end{equation}
Define the inner part of the absolute value on left-hand side of (\ref{37}) as $D_G(\bOmega,\bY)$. We know for $\bOmega$ and $\bOmega_1$,
\begin{equation}
\begin{split}
|\|\mathcal{P}_{\bV(\bOmega)}(\bY_{j\centerdot})\|^2-\|\mathcal{P}_{\bV(\bOmega_1)}(\bY_{j\centerdot})\|^2|&=\bY_{j\cdot}^T(\bP_{\bV(\bOmega)}-\bP_{\bV(\bOmega_1)})\bY_{j\cdot}\\
&\leq 2\sqrt{K(n-K)} ||\bV(\bOmega_1)^{\bot}\bV(\bOmega)^T||_F \|\bY_{j\cdot}\|^2
\end{split}
\end{equation}
and
\begin{equation}
|\mathbb{E}\|\mathcal{P}_{\bV(\bOmega)}(\bY_{j\centerdot})\|^2-\mathbb{E}\|\mathcal{P}_{\bV(\bOmega_1)}(\bY_{j\centerdot})\|^2|\leq 2\sqrt{K(n-K)} ||\bV(\bOmega_1)^{\bot}\bV(\bOmega)^T||_F \|\bOmega_{0,n}^T(\bB_0)_{j\cdot}\|^2.
\end{equation}
Thus
\begin{equation}\label{40}
\begin{split}
|D_G(\bOmega,\bY)-D_G(\bOmega_1,\bY)|&\leq 2\sqrt{K(n-K)} ||\bV(\bOmega_1)^{\bot}\bV(\bOmega)^T||_F\\
&\times\Big(\frac{1}{G} \sum_{j=1}^G  \frac{1}{2\sigma_j^2}\|\bY_{j\cdot}\|^2 +\frac{1}{G} \sum_{j=1}^G  \frac{1}{2\sigma_j^2}\|\bOmega_{0,n}^T(\bB_0)_{j\cdot}\|^2\Big).
\end{split}
\end{equation}
In order to apply the Kolmogorov's strong law of large number, we check the variance of $\frac{1}{2\sigma_j^2}\|\mathcal{P}_{\bV(\bOmega_1)}(\bY_{j\centerdot})\|^2$ and $  \frac{1}{2\sigma_j^2}\|\bY_{j\cdot}\|^2 $:
\begin{equation}
\text{Var}(\frac{1}{2\sigma_j^2}\|\mathcal{P}_{\bV(\bOmega_1)}(\bY_{j\centerdot})\|^2)=\frac{1}{\sigma_j^2}\|\bV(\bOmega_1)\bOmega_{0,n}^T(\bB_0)_{j\cdot}\|^2+K/2
\end{equation}
\begin{equation}
\text{Var}(\frac{1}{2\sigma_j^2}\|\bY_{j\centerdot}\|^2)=\frac{1}{\sigma_j^2}\|\bOmega_{0,n}^T(\bB_0)_{j\cdot}\|^2+n/2.
\end{equation}
Both of them are uniformly upper bounded with respect to $j$. So by Kolmogorov's strong law, we  have for every fixed $\bOmega_1$, $D_G(\bOmega_1,\bY)$ is almost surely converging to 0 as $G\rightarrow \infty$ and
\begin{equation}
\frac{1}{G} \sum_{j=1}^G  \frac{1}{2\sigma_j^2}\|Y_{j\cdot}\|^2 - \frac{1}{G} \sum_{j=1}^G  \Big(\frac{1}{2\sigma_j^2}\|\bOmega_{0,n}^T(\bB_0)_{j\cdot}\|^2+n/2\Big)\rightarrow 0 \ \ a.s.
\end{equation}
For a fixed $\epsilon>0$, define a neighborhood $U_{\bV(\bOmega_1)}$ for every $\bV(\bOmega_1)$,
\begin{equation}
\begin{split}
U_{\bV(\bOmega_1)}=\{\bV: ||\bV(\bOmega_1)^{\bot}\bV^T||_F<\frac{\epsilon}{4\sqrt{K(n-K)}} \Big(||\bOmega_{0,n}||^2_F \max_j\Big|\Big|\frac{(\bB_0)_{j\cdot}}{\sigma_j}\Big|\Big|^2+n/2+\epsilon\Big)^{-1},\\
\bV \text{ is an orthonormal $K$-frames in } \mathbb{R}^n \}
\end{split}
\end{equation}
Let $\mathcal{V}$ denote the Stiefel manifold $St(K,n)$, then there exists $\bOmega_1, \bOmega_2, \cdots, \bOmega_m$ such that $\mathcal{V}=\bigcup_{t=1}^m U_{\bV(\bOmega_t)}$. For $t=1,\cdots,m$, $D_G(\bOmega_t,\bY)\rightarrow 0$ almost surely, let $\mathcal{Y}$ denotes the realizations of $\bY$ such that $D_G(\bOmega_t,\bY)\rightarrow 0$ for all $t$ and 
\[\frac{1}{G} \sum_{j=1}^G  \frac{1}{2\sigma_j^2}\|\bY_{j\cdot}\|^2 - \frac{1}{G} \sum_{j=1}^G  \Big(\frac{1}{2\sigma_j^2}\|\bOmega_{0,n}^T(\bB_0)_{j\cdot}\|^2+n/2\Big)\rightarrow 0.\]
By definition $P(\mathcal{Y})=1$, for a realization $\by$ in $\mathcal{Y}$ there exist $G_0, G_1,\cdots,G_m$ such that
\begin{equation}\label{45}
\frac{1}{G} \sum_{j=1}^G  \frac{1}{2\sigma_j^2}\|\by_{j\cdot}\|^2 - \frac{1}{G} \sum_{j=1}^G  \Big(\frac{1}{2\sigma_j^2}\|\bOmega_{0,n}^T(\bB_0)_{j\cdot}\|^2+n/2\Big)<\epsilon/2, \text{for}\ G>G_0
\end{equation}
\begin{equation}\label{46}
D_G(\bOmega_t,\by)<\epsilon/2, \text{for}\ G>G_t, t=1,\cdots,m
\end{equation}
When $G>\max_{t}\{G_t\}$, for any $\bOmega$, there exists $\bOmega_{t_0}$ such that $\bV(\bOmega)\in U_{\bV(\bOmega_{t_0})}$, by (\ref{40}), (\ref{45}) and (\ref{46}):
\begin{equation}
|D_G(\bOmega,\by)|\leq |D_G(\bOmega,\by)-D_G(\bOmega_{t_0},\by)|+|D_G(\bOmega_{t_0},\by)|\leq \epsilon
\end{equation}
From here we have proved (\ref{37}).

Combined with lemma 5.1, we know when $G>\max_{t}\{G_t\}$,
\begin{equation}
\begin{split}
&P(||\bV(\bOmega_{0,n})^{\bot}\bV(\bOmega)^T||_F<\tilde{\epsilon}/L|\by, \bSigma_G) \\
=&C \int_{\{\bV(\bOmega):||\bV(\bOmega_{0,n})^{\bot}\bV(\bOmega)^T||_F<\tilde{\epsilon}/L\}} \text{exp}\Big(\sum_{j=1}^G\frac{1}{2\sigma_j^2}\|\mathcal{P}_{\bV(\bOmega)}(\by_{j\centerdot})\|^2\Big)m(d \bV(\bOmega))\\
=&\tilde{C} \int_{\{\bV(\bOmega):||\bV(\bOmega_{0,n})^{\bot}\bV(\bOmega)^T||_F<\tilde{\epsilon}/L\}} \text{exp}\Big(G\ D_G(\bOmega,\by)\\
&-\frac{1}{2}||\bV(\bOmega)^{\bot}\bV(\bOmega_{0,n})^T\bK(\bOmega_{0,n})^T\bB_{0,G}^T\bSigma_G^{-1/2}||_F^2\Big)m(d \bV(\bOmega))\\
\geq & \tilde{C} m_n(\{\bV: ||\bV_0 \bV^T||_F < \frac{\tilde{\epsilon}}{L}\}) \exp\Big(-\frac{1}{2}\frac{\tilde{\epsilon}^2}{L^2}\lambda_{max}(\bSigma_G^{-1/2}\bB_{0,G}\bK(\bOmega_{0,n}))^2-G\epsilon\Big),
\end{split}
\end{equation}
and on the other hand,
\begin{equation}
\begin{split}
&P(||\bV(\bOmega_{0,n})^{\bot}\bV(\bOmega)^T||_F>\tilde{\epsilon}|\by, \bSigma_G) \\
\leq & \tilde{C} m_n(\{\bV: ||\bV_0 \bV^T||_F >\tilde{\epsilon}\}) \exp\Big(-\frac{1}{2}\tilde{\epsilon}^2\lambda_{min}(\bSigma_G^{-1/2}\bB_{0,G}\bK(\bOmega_{0,n}))^2+G\epsilon\Big).
\end{split}
\end{equation}
Therefore we have 
\begin{equation}\label{50}
\begin{split}
\frac{P(||\bV(\bOmega_{0,n})^{\bot}\bV(\bOmega)^T||_F<\tilde{\epsilon}/L|\by, \bSigma_G)}{P(||\bV(\bOmega_{0,n})^{\bot}\bV(\bOmega)^T||_F>\tilde{\epsilon}|\by, \bSigma_G)}\geq &m_n(\{\bV: ||\bV_0 \bV^T||_F < \frac{\tilde{\epsilon}}{L}\})\\
&\times \exp\Big(\frac{3}{8}\tilde{\epsilon}^2 \lambda_{min}(\bSigma_G^{-1/2}\bB_{0,G}\bK(\bOmega_{0,n}))^2-2G\epsilon\Big)
\end{split}
\end{equation}
Since $\lambda_{\min}(\bB_{0,G})/\sqrt{G}$ is lower bounded,  $\lambda_{min}(\bSigma_G^{-1/2}\bB_{0,G}\bK(\bOmega_{0,n}))/\sqrt{G}$ is also lower bounded. Select $\epsilon$ such that
$$\epsilon\leq \frac{1}{8} \tilde{\epsilon}^2\Big(\lambda_{min}(\bSigma_G^{-1/2}\bB_{0,G}\bK(\bOmega_{0,n}))/\sqrt{G}\Big)^2,$$
then the right hand side of (\ref{50}) is no smaller than
$$m_n(\{\bV: ||\bV_0 \bV^T||_F < \frac{\tilde{\epsilon}}{L}\})\times \exp\Big(\frac{1}{8}\tilde{\epsilon}^2 \lambda_{min}(\bSigma_G^{-1/2}\bB_{0,G}\bK(\bOmega_{0,n}))^2\Big).$$
which goes to infinity by the lower boundedness of $\lambda_{min}(\bB_{0,G})/\sqrt{G}$.  

Thus 
$||\bV(\bOmega_{0,n})^{\bot}\bV(\bOmega)^T||_F|\by, \bSigma\rightarrow 0$ in probability for every $\by$ in $\mathcal{Y}$ which leads to the conclusion.

\end{proof}
The spirit of this proof is essentially the same as that of the classical Bayesian consistency theorem,
but is involved with infinite-dimensional potential data. 
In theorem 5.2, we made the assumption that the $L_2$ norm of rows of $\bB_0$ are upper bounded due to the proof, which restricted ourselves to the case where all singular values of $\bB_{0,G}$ are increasing at the order of $\sqrt{G}$.
This condition can be satisfied when rows of $\bB_0$ are i.i.d from an underlying distribution $p_B$: 
\[\lambda_{k}(\bB_{0,G})/\sqrt{G}=\sqrt{\lambda_{k}(\bB_{0,G}^T\bB_{0,G}/G)}\rightarrow \sqrt{\lambda_{k}(E_{p_B}(\bB_{j\cdot}\bB_{j\cdot}^T))}\ ,\ G \rightarrow \infty \ \ a.s.\]

\subsection{Remark of Section 5.1.2}
From \cite{cai2018rate}, for every pair of $\bV(\bOmega_{0,n})$ and $\bV(\bOmega)$ there exists an orthogonal matrix $\bW$ such that $||\bV(\bOmega)-\bW \bV(\bOmega_{0,n})||_F \leq \sqrt{2}|| \sin(\angle(\bV(\bOmega_{0,n}),\bV(\bOmega)))||_F$ where $\angle(\bV(\bOmega_{0,n}),\bV(\bOmega))$ denotes the diagonal matrix formed by canonical angles between row spaces of $\bOmega_{0,n}$ and $\bOmega$. For fixed $n$ and $G=s \rightarrow \infty$, using the shrinkage of canonical angles between row spaces from Theorem 5.2, there exists a orthogonal random matrix $\bW$ such that $||\bV(\bOmega)-\bW \bV(\bOmega_{0,n})||_F|\bY,\bSigma \rightarrow 0$ in probability as $G \to \infty$. The posterior distribution of $\bV(\bOmega)$ conditioned on the row vector space of $\bOmega$ is actually an uniform distribution on all the orthonormal basis within since the density in (9) involves $\bV(\bOmega)$ only through the row vector space. Therefore $\bV(\bOmega)|\bY,\bSigma \sim \bO_1 \bV(\bOmega)|\bY,\bSigma\sim \bO_1 (\bW\bV(\bOmega_{0,n})+(\bV(\bOmega)-\bW\bV(\bOmega_{0,n})))|\bY,\bSigma$ for an independent uniform random orthogonal matrix $\bO_1$. Since $||\bO_1(\bV(\bOmega)-\bW\bV(\bOmega_{0,n}))||_F|\bY,\bSigma \rightarrow 0$, the posterior sample of $\bV(\bOmega)$ can be asymptotically express as $\bO \bV(\bOmega_{0,n})$ where $\bO=\bO_1\bW$ is an independent uniform random orthogonal matrix, i.e., $\bV(\bOmega)$ differs $\bO \bV(\bOmega_{0,n})$ by a matrix that has Frobenius norm converging to 0 under the asymptotic regime of Theorem 5.2.

\subsection{Proof of Theorem 5.3}
Theorem 5.3 is an immediate result of the following lemma and  Theorem 5.2.

\begin{lemma}
Let $(\bB_0,\bGamma_0)$ be a regular infinite loading pair with $\bGamma_0$ known, $\bOmega_0$ be a $K \times \infty$ matrix and $\bSigma=diag(\sigma_1^2,\cdots)$ be a known infinite diagonal matrix.  Define $\bSigma_G=diag(\sigma^2_{\pi^{-1}(1)},\cdots,\sigma^2_{\pi^{-1}(G)})$ and $\bSigma_{G}^{(k)}=diag(\sigma^2_{\pi^{-1}(l_{0,k})},\cdots,\sigma^2_{\pi^{-1}(l_{0,k+1}-1)})$. $\bOmega_{0,n}$ denotes the matrix formed by the first $n$ columns of $\bOmega$. Suppose there exists an $\epsilon>0$ such that the following holds for the increasing pair $(n,G)=\{(n_t,G_t)\}_{t=1,\cdots}$.

1. $\min_{k'} \lambda_{min}((\bSigma_G^{(k')})^{-1/2}\bB_{0,G}^{(k')}\bK(\bOmega_{0,n})_{1:k'})\rightarrow \infty$ as $t\rightarrow \infty$.

2. 
Let $\bV_0$ be any fixed $K \times n$ orthonormal matrix,
\[
\begin{split}
-\text{log}(m_n(\bigcap_{k=1}^K\Big\{\bV:& ||(\bV_0)_{1:k}^{\bot} \bV_{1:k}^T||_F< \frac{\epsilon \min_{k'} \lambda_{min}((\bSigma_G^{(k')})^{-1/2}\bB_{0,G}^{(k')}\bK(\bOmega_{0,n})_{1:k'})}{\lambda_{max}((\bSigma_G^{(k)})^{-1/2}\bB_{0,G}^{(k)}\bK(\bOmega_{0,n})_{1:k})}\Big\}))\\
=&o(\epsilon^2 \min_{k'} \lambda_{min}((\bSigma_G^{(k')})^{-1/2}\bB_{0,G}^{(k')}\bK(\bOmega_{0,n})_{1:k'})^2)\ \text{as}\ t\rightarrow \infty.
\end{split}
\]

Let $\bY=\bB_{0,G} \bOmega_{0,n}$ and model $\bY_{\cdot i}$ with $\mathcal{N}_G(\bB\bOmega_{\cdot i},\bSigma_G)$ for $i=1,\cdots,n$. Impose a point mass and flat mixture prior on entries of $\bB$ according to the feature allocation matrix $\bGamma_{0,G}$ and assume a distribution on $\bOmega$ that is invariant under right orthogonal transformations, then for a random draw $\bOmega$ from its posterior distribution, $$P(\bigcup_{k=1}^K \{\bV: ||\bV(\bOmega_{0,n})_{1:k}^{\bot}\bV(\bOmega)_{1:k}^T||_F>\sqrt{K+1}\epsilon\}|\bY, \bSigma_G,\bGamma_{0,G})\rightarrow 0$$ as $t \rightarrow \infty$.
\end{lemma}

\begin{proof}
We know that for $f(n,G)=\epsilon \min_{k'} \lambda_{min}((\bSigma_G^{(k')})^{-1/2}\bB_{0,G}^{(k')}\bK(\bOmega_{0,n})_{1:k'})$:

1'. $f(n,G)$ goes to infinity.

2'. Let $\bV_0$ be a fixed $K \times n$ orthonormal matrix, 
$$-\text{log}(m_n(\bigcap_{k=1}^K\Big\{\bV: ||(\bV_0)_{1:k}^{\bot} \bV_{1:k}^T||_F< \frac{f(n,G)}{\lambda_{max}((\bSigma_G^{(k)})^{-1/2}\bB_{0,G}^{(k)}\bK(\bOmega_{0,n})_{1:k})}\Big\}))=o(f(n,G)^2).$$
Define two disjoint set $S_1$ and $S_2$ as following $$S_1=\bigcap_{k=1}^K\Big\{\bV: ||\bV(\bOmega_{0,n})_{1:k}^{\bot} \bV_{1:k}^T||_F< \frac{f(n,G)}{\lambda_{max}((\bSigma_G^{(k)})^{-1/2}\bB_{0,G}^{(k)}\bK(\bOmega_{0,n})_{1:k})}\Big\}$$  
$$S_2= \bigcup_{k=1}^K\Big\{\bV:||\bV(\bOmega_{0,n})_{1:k}^{\bot}\bV_{1:k}^T||_F>\frac{\sqrt{K+1}f(n,G)}{\lambda_{min}((\bSigma_G^{(k)})^{-1/2}\bB_{0,G}^{(k)}\bK(\bOmega_{0,n})_{1:k})}\Big\}$$
Similar as (\ref{35}), we can compute:
\begin{equation}\label{51}
\begin{split}
&P(\bV(\bOmega) \in S_1|\bY,\bSigma_G,\bGamma_{0,G})\\
=& C \int_{S_1} \exp\Big(-\frac{1}{2}\sum_{k=1}^K ||\bV(\bOmega)_{1:k}^{\bot}\bV(\bOmega_{0,n})_{1:k}^T\bK(\bOmega_{0,n})_{1:k}^T(\bB_{0,G}^{(k)})^T(\bSigma_G^{(k)})^{-1/2}||_F^2\Big)m(d \bV(\bOmega))\\
\geq& C \cdot m_n(S_1) \exp(-\frac{K}{2}f(n,G)^2)
\end{split}
\end{equation}
\begin{equation}\label{52}
\begin{split}
&P(\bV(\bOmega) \in S_2|\bY,\bSigma_G,\bGamma_{0,G})\\
=& C \int_{S_2} \exp\Big(-\frac{1}{2}\sum_{k=1}^K ||\bV(\bOmega)_{1:k}^{\bot}\bV(\bOmega_{0,n})_{1:k}^T\bK(\bOmega_{0,n})_{1:k}^T(\bB_{0,G}^{(k)})^T(\bSigma_G^{(k)})^{-1/2}||_F^2\Big)m(d \bV(\bOmega))\\
\leq& C \cdot m_n(S_2) \exp(-\frac{K+1}{2}f(n,G)^2)
\end{split}
\end{equation}Combine (\ref{51}) and (\ref{52}), we have:
\begin{equation}
\frac{P(\bV(\bOmega) \in S_1|\bY,\bSigma_G,\bGamma_{0,G})}{P(\bV(\bOmega) \in S_2|\bY,\bSigma_G,\bGamma_{0,G})} \geq m_n(S_1) \exp(\frac{1}{2}f(n,G)^2)
\end{equation}
From condition 2', the right hand side goes to infinity for the increasing pair $(n,G)=\{(n_t,G_t)\}_{t=1,\cdots}$ as $t \to \infty$, thus
$$P(\bV(\bOmega) \in S_2|\bY,\bSigma_G,\bGamma_{0,G})\rightarrow 0.$$
Therefore, $$P(\bigcup_{k=1}^K \{ \bV: ||\bV(\bOmega_{0,n})_{1:k}^{\bot}\bV(\bOmega)_{1:k}^T||_F>\sqrt{K+1}\epsilon\}|\bY, \bSigma_G,\bGamma_{0,G})\rightarrow 0 \text{ as } t \to \infty.$$
\end{proof}

\section{Additional figures}
\subsection{The AGEMAP dataset}

\begin{figure}[!htb]
     \centering
     \subfloat[The SpSL-orthonormal factor model]{\includegraphics[width=0.74\textwidth]{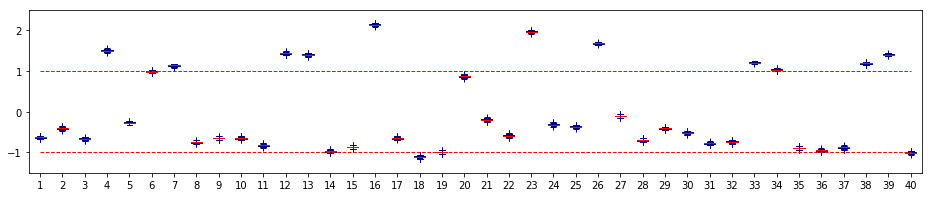}}

     \subfloat[The modified Ghosh-Dunson model]{\includegraphics[width=0.74\textwidth]{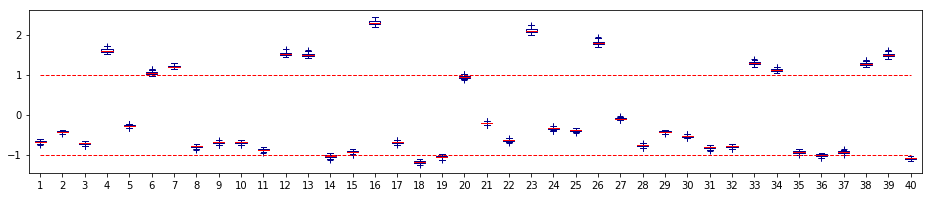}}
     \caption{Boxplots of posterior samples of the latent factors under specified models.}
\end{figure}

\begin{figure}[!htb]
     \centering
     \subfloat[The SpSL-orthonormal factor model]{\includegraphics[width=0.74\textwidth]{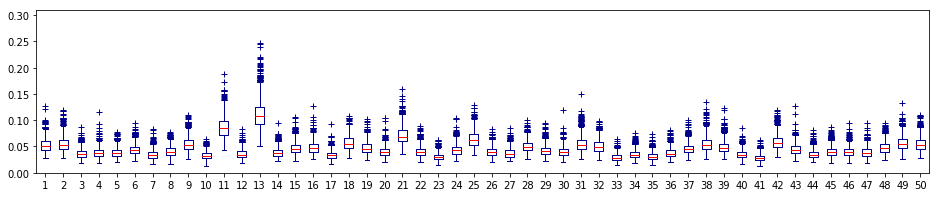}}

     \subfloat[The modified Ghosh-Dunson model]{\includegraphics[width=0.74\textwidth]{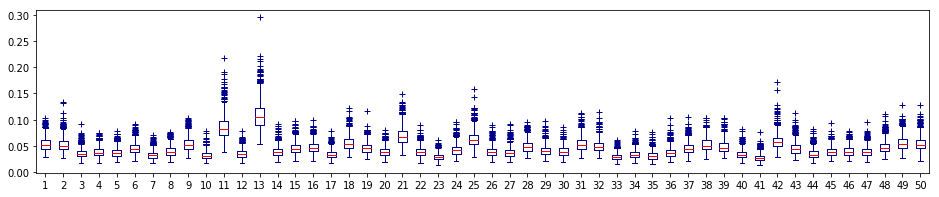}}
     \caption{Boxplots of posterior samples of the first 50 entries of idiosyncratic variances under specified models.}
\end{figure}

\begin{figure}[!htb]
     \centering
     \subfloat[The SpSL-orthonormal factor model]{\includegraphics[width=0.74\textwidth]{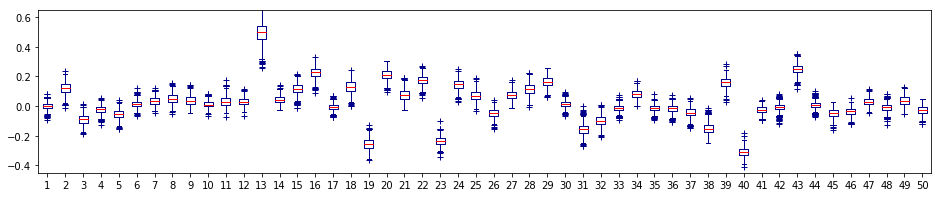}}

     \subfloat[The modified Ghosh-Dunson model]{\includegraphics[width=0.74\textwidth]{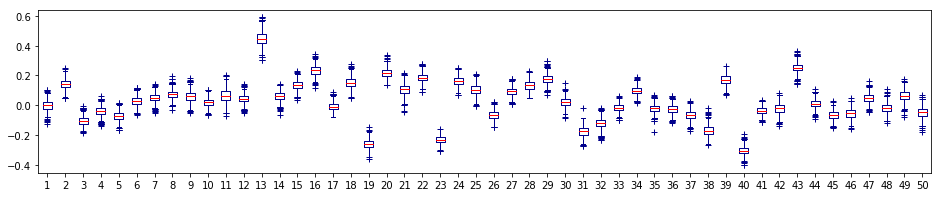}}
     \caption{Boxplots of posterior samples of the first 50 entries of the loading vector under specified models.}
\end{figure}


\newpage
\subsection{The synthetic example}
\begin{figure}[!htb]
\centering
\includegraphics[width=0.85\textwidth]{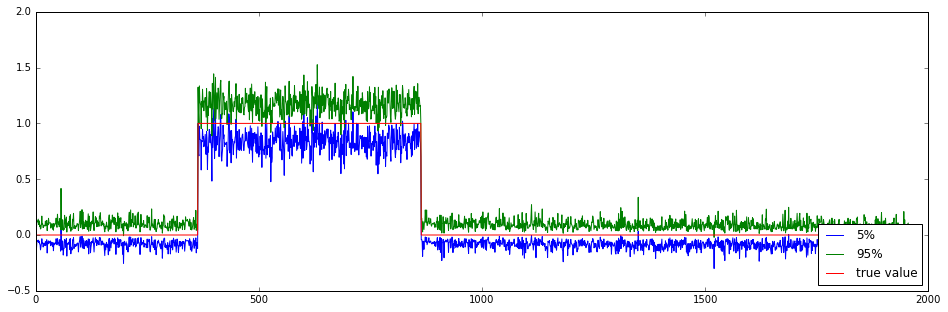}  
\includegraphics[width=0.85\textwidth]{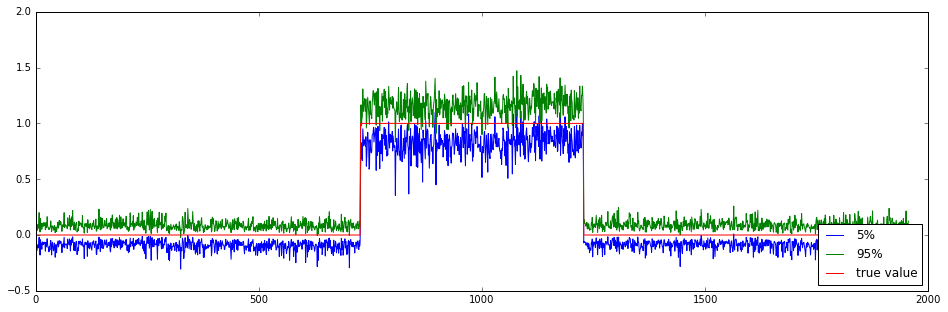}  
\includegraphics[width=0.85\textwidth]{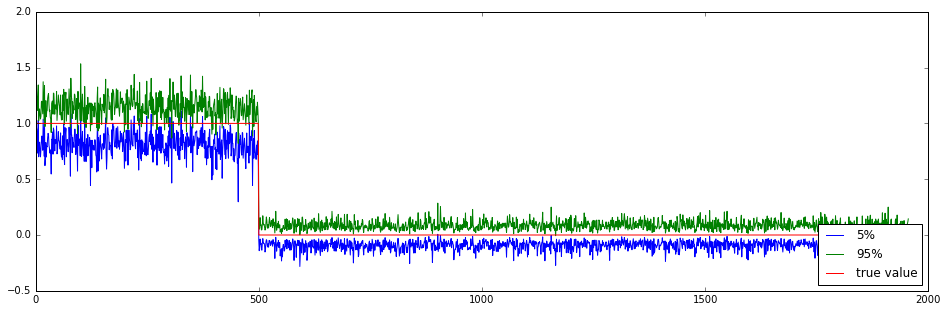}  
\includegraphics[width=0.85\textwidth]{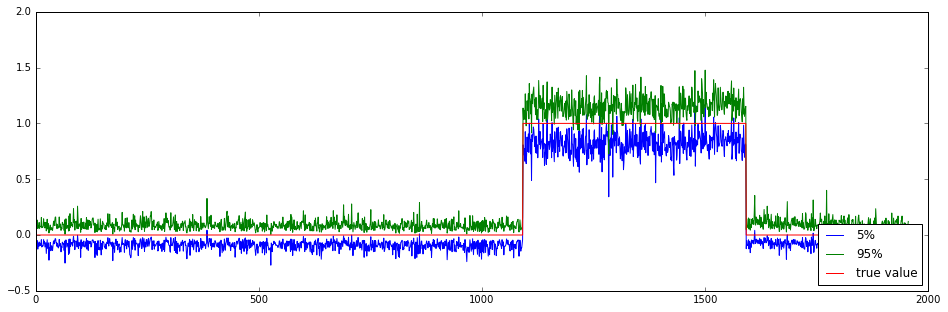}  
\includegraphics[width=0.85\textwidth]{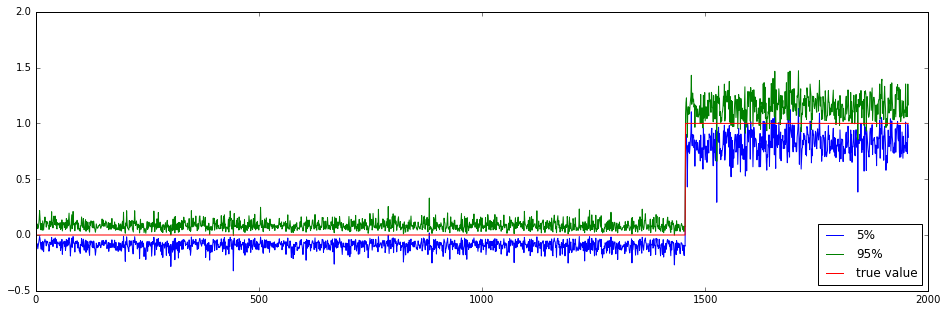}  
\caption{The $90\%$ credible intervals for elements in the first five columns of the loading matrix under $\sqrt{n}$-orthonormal factor model for the synthetic example, $\lambda_0=20,\ \lambda_1=0.001.$}
\end{figure}
\clearpage
\begin{figure}[!htb]
\centering
\includegraphics[width=0.85\textwidth]{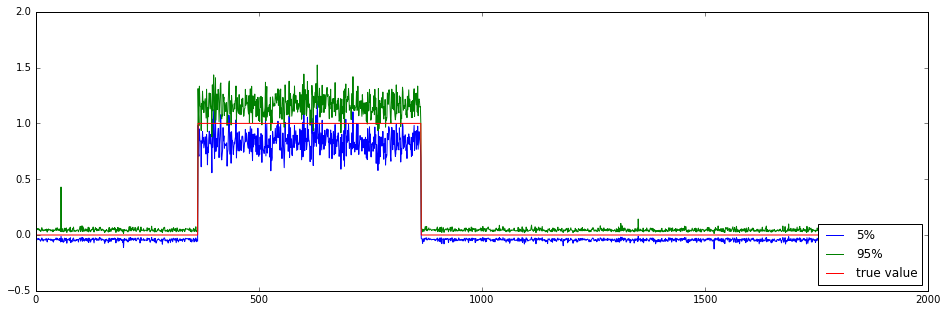} 
\includegraphics[width=0.85\textwidth]{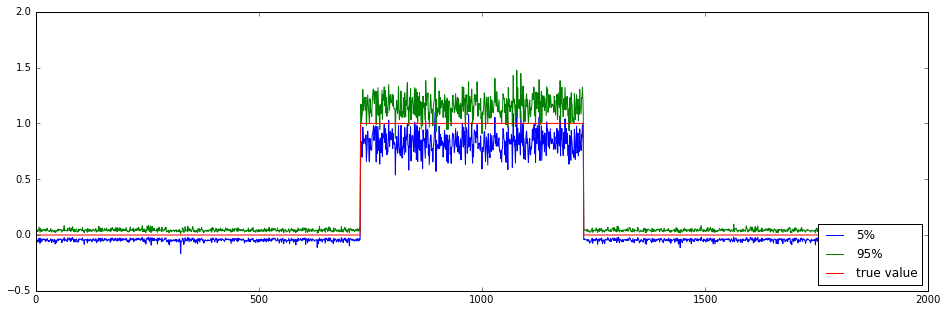}  
\includegraphics[width=0.85\textwidth]{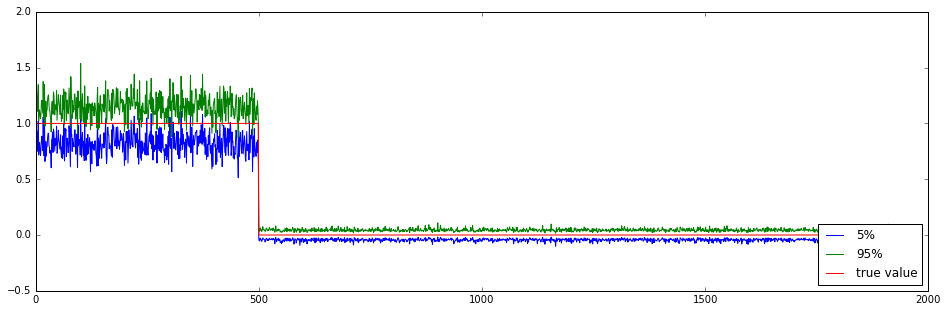}  
\includegraphics[width=0.85\textwidth]{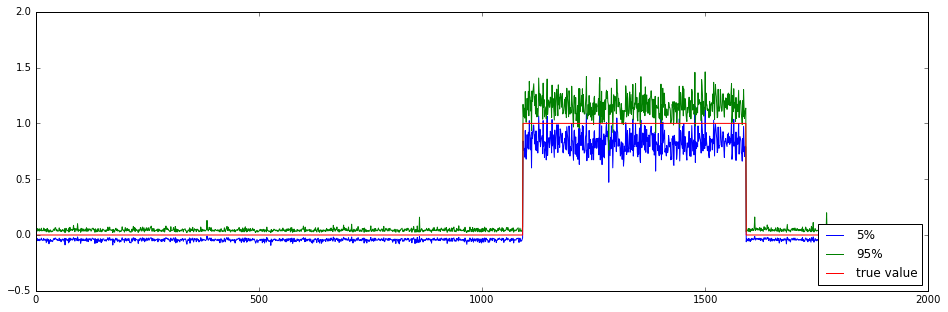}  
\includegraphics[width=0.85\textwidth]{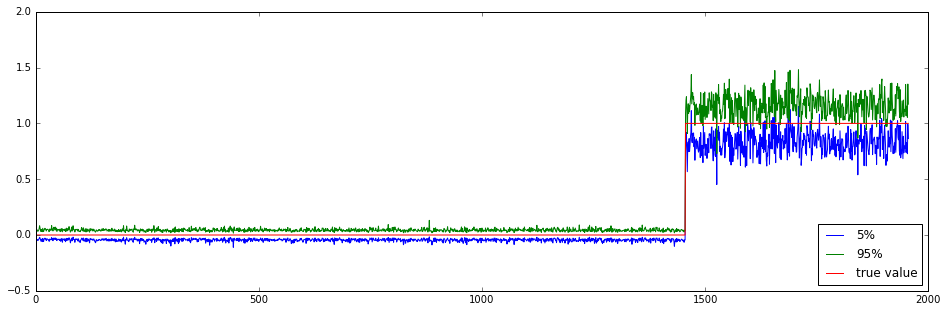}  
\caption{The $90\%$ credible intervals for elements in the first five columns of the loading matrix under the $\sqrt{n}$-orthonormal factor model for the synthetic example, $\lambda_0=50,\ \lambda_1=0.001.$}
\end{figure}
\clearpage
\begin{figure}[!htb]
\centering
\includegraphics[width=0.85\textwidth]{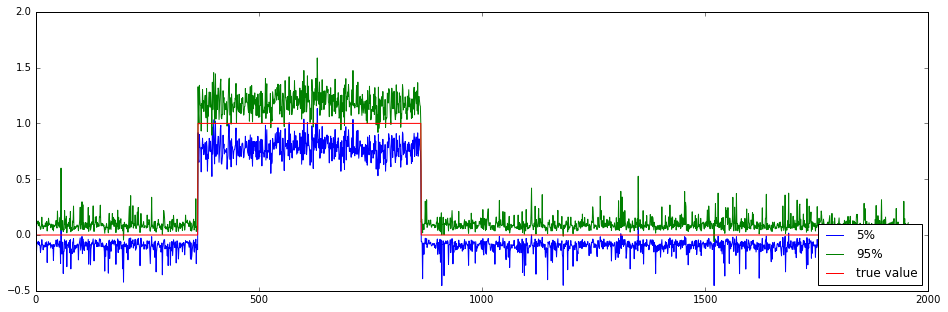}  
\includegraphics[width=0.85\textwidth]{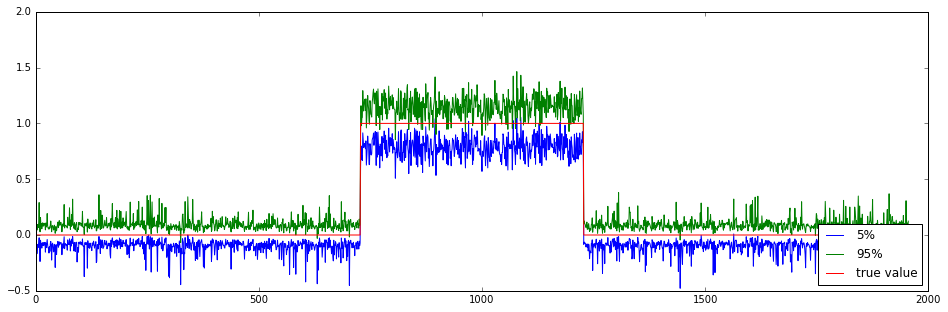}  
\includegraphics[width=0.85\textwidth]{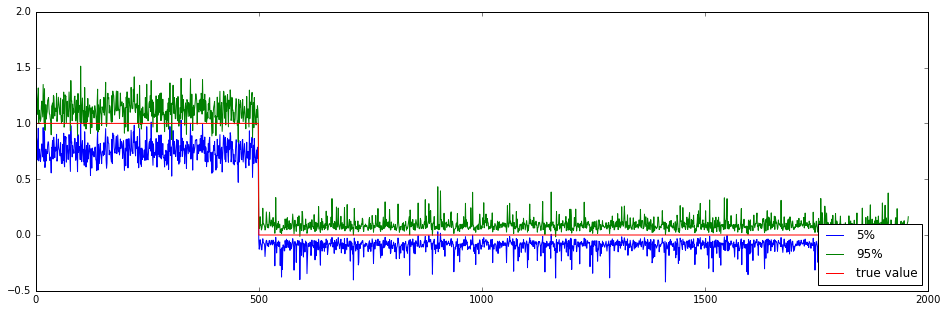}  
\includegraphics[width=0.85\textwidth]{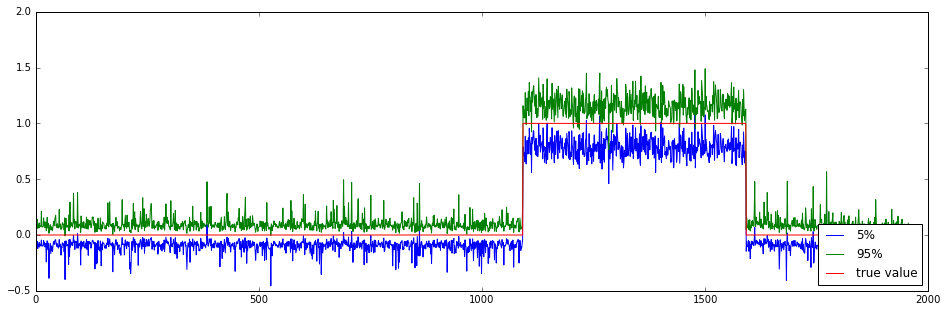}  
\includegraphics[width=0.85\textwidth]{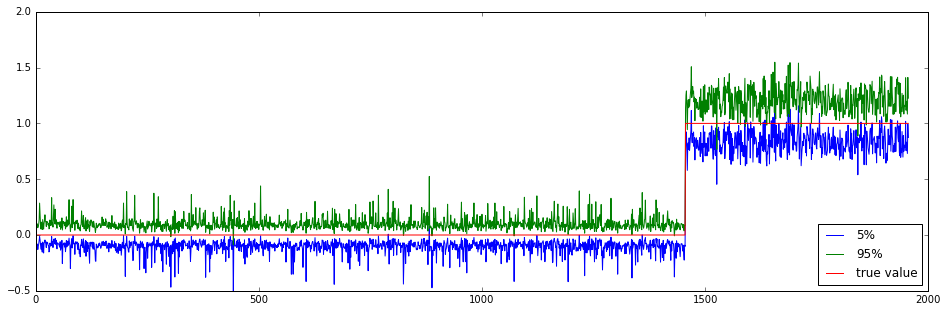}  
\caption{The $90\%$ credible intervals for elements in the first five columns of the loading matrix using modified Ghosh-Dunson model for the synthetic example.}
\end{figure}
\clearpage
\begin{figure}[!htb]
\centering
\includegraphics[width=0.85\textwidth]{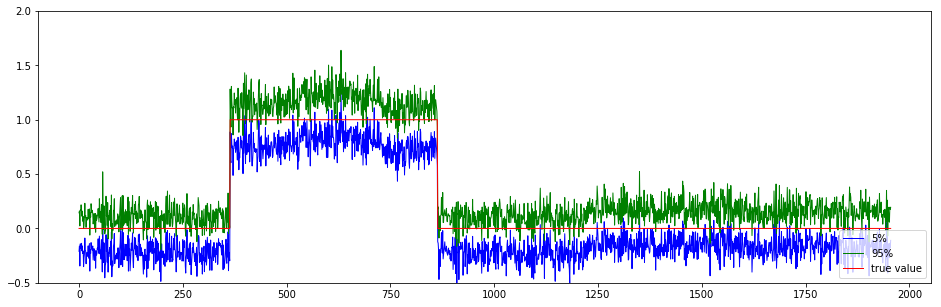}  
\includegraphics[width=0.85\textwidth]{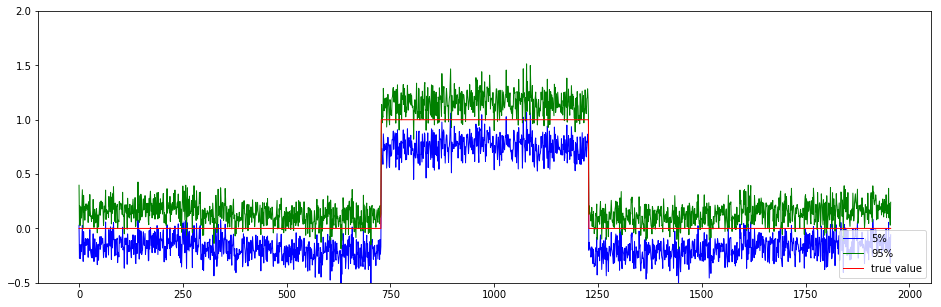}  
\includegraphics[width=0.85\textwidth]{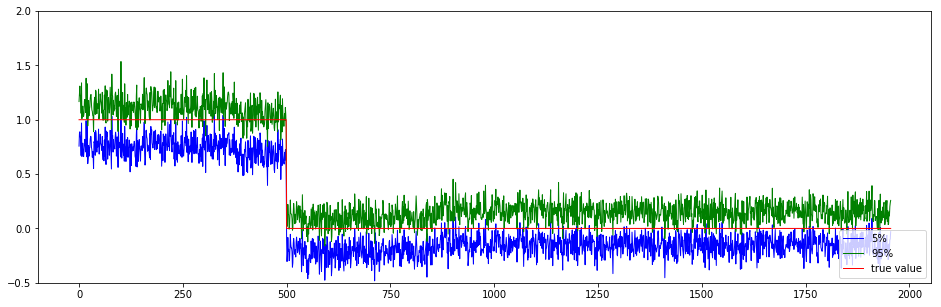}  
\includegraphics[width=0.85\textwidth]{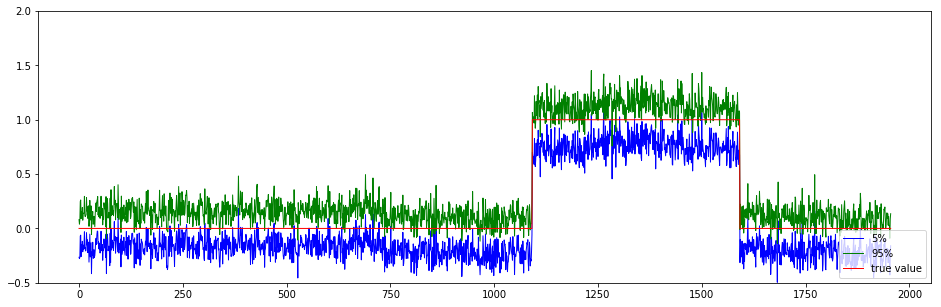}  
\includegraphics[width=0.85\textwidth]{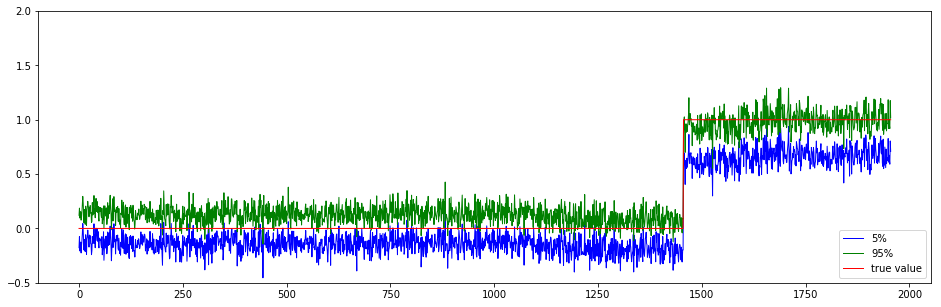}  
\caption{The $90\%$ credible interval for elements in the first five columns of the loading matrix under the  model from \cite{Bhattacharya2011Sparse} for the synthetic example.}
\end{figure}

\clearpage
\begin{figure}[!htb]
\centering
\includegraphics[width=0.85\textwidth]{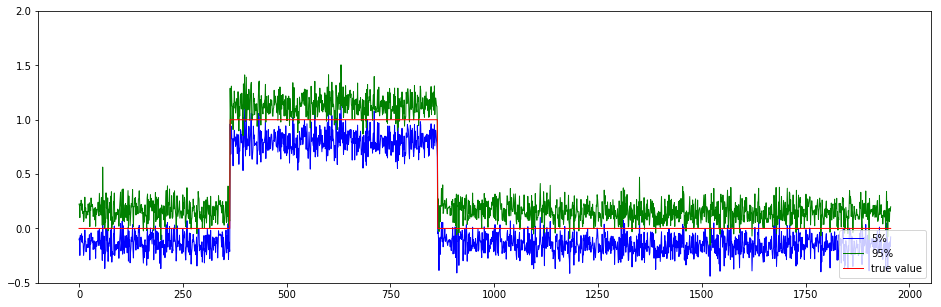}  
\includegraphics[width=0.85\textwidth]{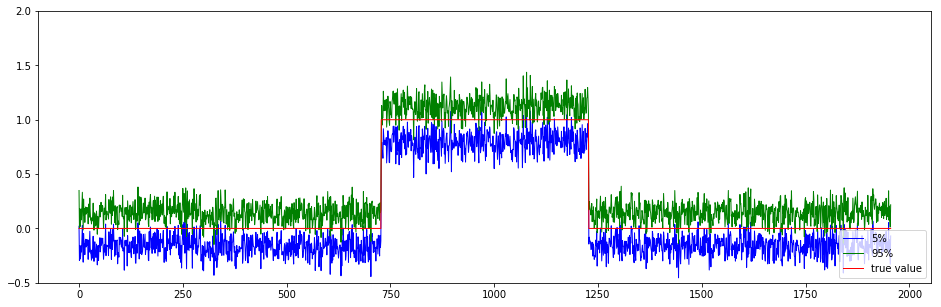}  
\includegraphics[width=0.85\textwidth]{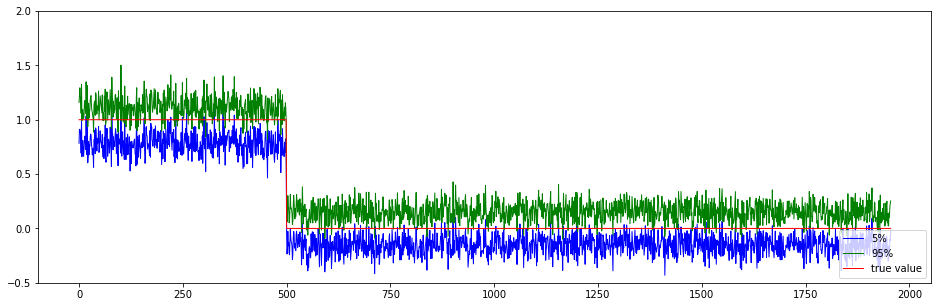}  
\includegraphics[width=0.85\textwidth]{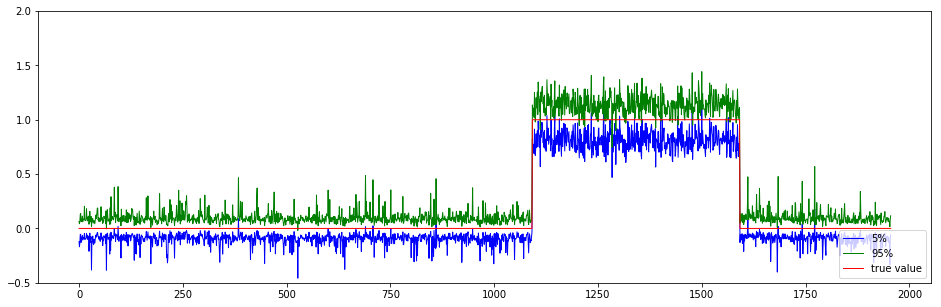}  
\includegraphics[width=0.85\textwidth]{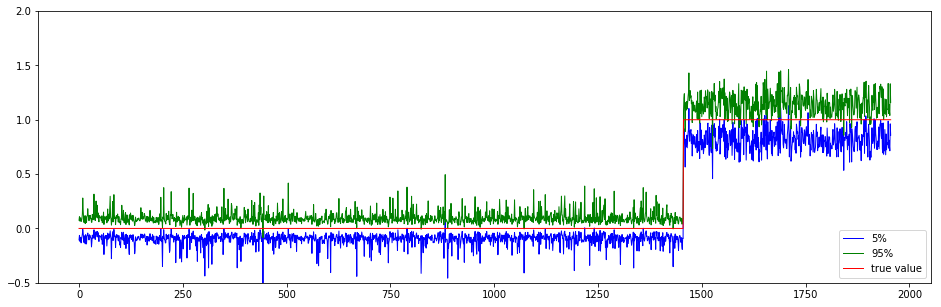}  
\caption{The $90\%$ credible interval for elements in the first five columns of the loading matrix under the modified Ghosh-Dunson model with $\sqrt{n}$-orthonormal factors.}
\end{figure}

\clearpage
\begin{figure}[!htb]
\centering
\includegraphics[width=0.85\textwidth]{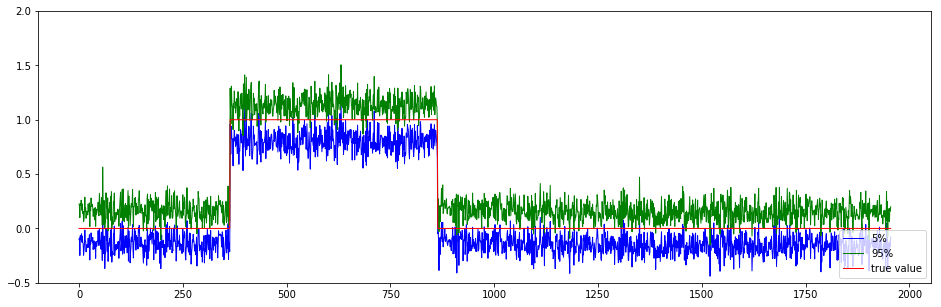}  
\includegraphics[width=0.85\textwidth]{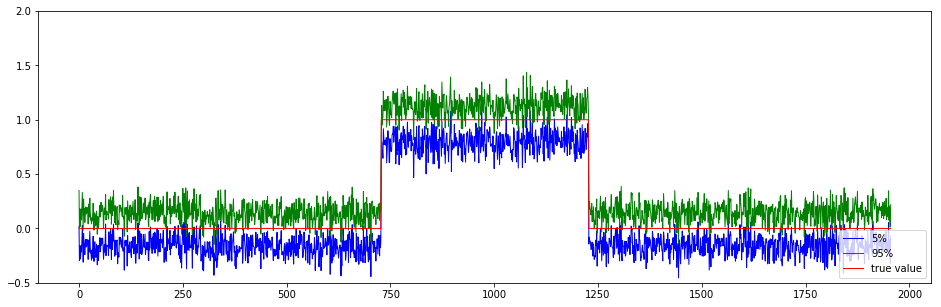}  
\includegraphics[width=0.85\textwidth]{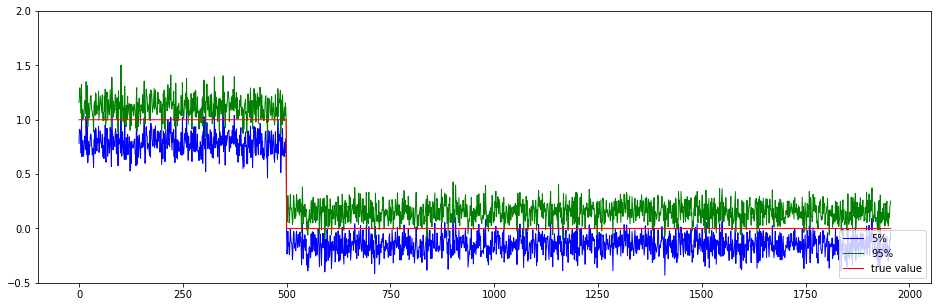}  
\includegraphics[width=0.85\textwidth]{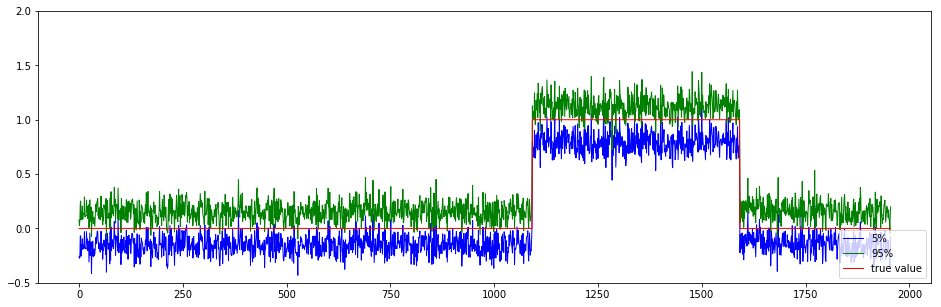}  
\includegraphics[width=0.85\textwidth]{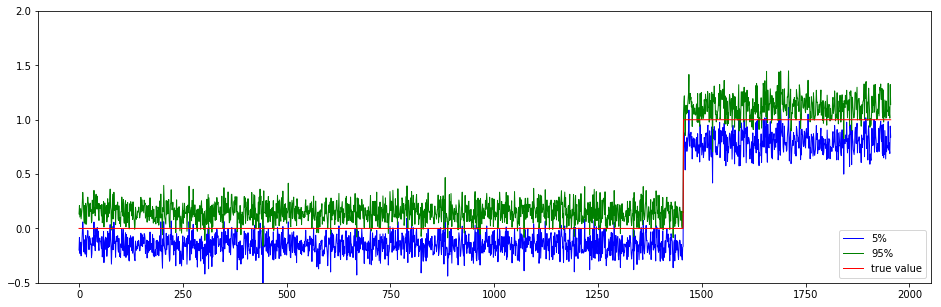}  
\caption{$90\%$ credible interval for elements in first five columns of loading matrix using model from \cite{Bhattacharya2011Sparse} with $\sqrt{n}$-orthonormal factors.}
\end{figure}
\end{appendix}
\end{document}